\documentclass[letter,11pt]{article}
\usepackage{amsmath,amssymb,amsthm}
\usepackage{graphicx}
\usepackage{latexsym}
\usepackage{paralist}          
\usepackage[margin=1in]{geometry}   
\usepackage{subfig}  
\usepackage{indentfirst}
\usepackage{color, soul}       
\usepackage{url}

 \newtheorem{theorem}{Theorem}[section]
 \newtheorem{proposition}[theorem]{Proposition}
 \newtheorem{corollary}[theorem]{Corollary}

 \newtheorem{lemma}[theorem]{Lemma}
 \newtheorem{claim}[theorem]{Claim}

\newtheorem{definition}{Definition}

\newcommand{\alg}[1]{\textsc{#1}}
\newcommand{\br}{\alg{Bcast}}

\newcommand{\lbr}{\alg{LocalBcast}}

\newcommand{\ta}{\alg{Try\&Adjust}}

\def \cI {\mathcal{I}}
\def \hcI {\hat{\mathcal{I}}}
\def \bcI {\hat{\mathcal{I}}}
\def \hphi {\hat{\eta}}

\def \Icd {\mathcal{I}_{cd}}
\def \Ic {\mathcal{I}_{c}}
\def \Iack {\mathcal{I}_{ack}}

\def \e {\epsilon}
\def \z {\zeta}
\def \l {\lambda}
\def \R {R}
\def \rmin {r_{min}}

\def \scc {\textsf{SuccClear}}

\newcommand{\prim}[1]{\textsf{#1}}
\def \cd {\prim{CD}}
\def \ack {\prim{ACK}}
\def \ntd {\prim{NTD}}

\newcommand{\busy}{\textsc{Busy}}
\newcommand{\idle}{\textsc{Idle}}

\newcommand{\mypar}[1]{\smallskip\noindent\textbf{#1}.\ }
\newcommand{\thead}[1]{\smallskip\noindent\emph{#1}.\ }

\begin{document}

\begin{titlepage}

\author{Magn\'us M. Halld\'orsson\thanks{ICE-TCS, School of Computer Science, Reykjavik University, Iceland; mmh@ru.is, ttonoyan@gmail.com.} \and
Tigran Tonoyan$^\ast$ \and
Yuexuan Wang\thanks{College of Comp.\ Sci. \& Tech., Zhejiang University, Hangzhou, P.R.\ China; Department of Computer Science, The University of Hong Kong, Hong Kong, P.R.\ China;  amywang@hku.hk.} 
\and Dongxiao Yu\thanks{Services Computing Technology and System Lab, Cluster and Grid Computing Lab in the School of Computer Science and Technology, Huazhong University of Science and Technology, 1037 Luoyu Road, Wuhan 430074, P.R. China; dxyu@hust.edu.cn.}
}

\date{\today}

\title{Data Dissemination in Unified Dynamic Wireless Networks
}

\maketitle
\thispagestyle{empty}

\begin{abstract}
We give efficient algorithms for the fundamental problems of Broadcast and Local Broadcast in dynamic wireless networks.
We propose a general model of communication which captures and includes both fading models (like SINR) and
graph-based models (such as quasi unit disc graphs, bounded-independence graphs, and protocol model).
The only requirement is that the nodes can be embedded in a bounded growth quasi-metric, which is the weakest condition known to ensure distributed operability.
Both the nodes and the links of the network are dynamic: nodes can come and go, while the signal strength on links can go up or down.

The results improve some of the known bounds even in the static setting, including an optimal algorithm for local broadcasting in the SINR model, which is additionally uniform (independent of network size).
An essential component is a procedure for balancing contention, which has potentially wide applicability.
The results illustrate the importance of carrier sensing, a stock feature of wireless nodes today,
which we encapsulate in primitives to better explore its uses and usefulness.
\end{abstract}
\end{titlepage}

\section{Introduction}


Wireless networks are ubiquitous and are on their way to become even more prevalent, e.g., with the advent of Internet-of-Things. Wireless communication is, however, particularly challenging to model algorithmically.
In two crucial interrelated aspects, wireless networks on the ground 
differ from models typically assumed in algorithmic
studies.  One is the \emph{communication modeling}: when is a transmission successfully decoded, as a function of the
environment and the \emph{interference} from other transmissions.  The other is \emph{variability
  with time}: wireless networks are particularly susceptible to changes.
Both of these are hard to capture accurately with well-defined, clear-cut rules. 

We aim in this paper to address core information dissemination problems -- \emph{local broadcast} and (global)
\emph{broadcast} -- in dynamic distributed networks, under very weak assumptions on the communication.
To
this end, we propose a communication model with significant flexibility that allows for adversarial control,
generalizing essentially all known analytic wireless models. It allows for more general interference relationships than
treated before.  The network can experience adversarial dynamic behavior, both edge changes (change of signal strengths)
and node insertions/deletions.

Wireless communication has traditionally been modeled theoretically by graphs, either geometric or general.
Interference is then also transmitted on graph edges (precluding a node from receiving a message from a neighbor if another neighbor is transmitting),
but is sometimes represented by a supergraph.
\emph{Fading channel} or \emph{physical} models common in communication engineering, on the other hand, 
consider interference as cumulative, decreasing with distance but adding up.
They have been popular in recent algorithmic studies, adding more realism to the formulation.
The standard assumption of \emph{geometric signal decay} (that signal decreases inverse polynomially with distance) in the SINR model, the prototypical fading channel,
is though equally at odds with experimental evidence.
Ultimately, it may prove futile to hope for a clean deterministic model, or even a purely stochastic one,
without a significant dose of unpredictability and non-determinism.

Wireless communication is commonly closely linked to mobility, as the transceivers are more often than not on the move.
Dynamic changes to reception conditions have many causes other than node mobility, since 
almost any changes in the environment affect transmissions due to reflections of signals over multiple paths,
antenna characteristics, scattering, and diffraction.
These changes are by nature hard to predict, even when assuming a ``mobility model''.
The most robust approach would then be to assume a non-trivial adversarial component.
While the study of algorithms in dynamic networks has a long history, little has been done in cases where interference plays a role.


\mypar{Setting and Model}
%
Nodes are distributed and autonomous. 
There is no built-in structure and the nodes have no information besides bounds on model parameters and an upper bound on the number of nodes.
Communication is locally synchronous, but there is no global clock. 

We will assume a very general model for when nodes successfully communicate.
Nodes are located in space, with separation between points given by the 
relative decrease in strength of signal (or interference) sent between the points.
This induces a metric space, when these decays are raised to the appropriate power -- actually, it is a \emph{quasi-metric}, since symmetry need not hold.
For distributed computation to be possible, the quasi-metric must have \emph{bounded independence} (to be defined precisely).
\emph{Edge changes} can occur, with some restrictions, 
which are changes in the signal strength between the pair of points.

The rule for when communication is successful is only partially pre-specified: transmission succeeds on a ``clear channel''.
That is, if a sender is within a \emph{communication radius} from the receiver,
if no other node transmits within a (larger) radius, and if the combined interference from all other transmitting nodes is (quite) small, 
then the transmission succeeds. Otherwise, success is up to the adversary, or it can be further specified by the particular model assumptions desired.
This captures essentially all known algorithmic wireless models (including quasi-unit disc graphs, unit-ball graphs, bounded-independence graphs, $k$-hop extensions, and SINR). The only exception is the radio network model with general graphs, which cannot be extended to involve comprehensive interference without a major hit in time complexity.

The generality of our model is a key feature. Given the vagaries of actual wireless environments, it is preferable for robustness reasons to make minimal assumptions about the communication model. A conservative approach is then to seek algorithms that work in most established models rather than depending on model-specific factors that simplify the life of the algorithm designer.

\mypar{Our Approach and Results}
%
The key algorithmic technique is a natural randomized \emph{contention balancing} procedure, where a node continuously adjusts its transmission probability based on the interference that it senses. It allows nodes to stabilize quickly from any initial conditions, or after waking up.
This routine is a variation on an old story, a simple backoff procedure to manage local contention:
\begin{quote}
  if a node $v$ senses contention in a given round beyond a fixed threshold, then $v$ halves its transmission probability in the next round and otherwise doubles it.
\end{quote}
Our main technical contribution is to show logarithmic-round convergence of this method to a steady state of nearly balanced contention, from an arbitrary starting configuration and in the presence of network changes.
It proves also to be surprisingly tolerant of different communication models.
The higher level algorithms are then built on top of this primitive.

A crucial component is the use of \emph{carrier sense} to detect the cumulative amount of signals in the air.
Since it is supplied by the cheapest available hardware today as RSS (received signal strength) readings, we posit that carrier-sense
capability should be the default assumption in wireless algorithmics
 (while exploring the necessity of different assumptions is interesting theoretically).
As carrier-sense indicators can provide fine-grained information, we are interested in restraining its use and identifying which aspects are necessary to 
achieve the results obtained. To this end, we identify several \emph{primitives} that carrier-sense can supply, restrict the algorithm to use only a subset of the primitives, and examine which of these are truly necessary.

The local broadcast algorithm simply runs the contention balancing procedure, with nodes bowing out when they are sure to have completed their transmission.
The broadcast algorithms are based on sparsifying the instance, so that only nodes of constant density actually participate in the global broadcast action.
The former is achieved in $O(\Delta + \log n)$ time, where $\Delta$ is the maximum number of neighbors that a node can have,
while the latter takes $O(\hat{D})$ rounds, where $\hat{D}$ is a dynamic diameter.

These dissemination algorithms are efficient enough to improve on some of the results known for static versions of the problems.
The local broadcast algorithm is strongly optimal, or within constant factors on every instance.
In the standard setting (static, spontaneous case), the algorithm is  \emph{uniform}, in that it need not know the network size.
The broadcast algorithm is also optimal and uniform in the same setting, 
while in the non-spontaneous setting it is faster by a logarithmic factor than the previous algorithm of \cite{JKRS14} that
however does not require carrier sense.



\mypar{Closely Related Work} 
There are two largely disjoint bodies of work of wireless algorithmic results, with work on fading models like SINR
slowly catching up with the better studied graph-based models.  One approach for capturing more realism in SINR model is
to move beyond Euclidean metrics \cite{FKV11}, even to general ones \cite{KV10}.  One can view relative signal decrease
as implicitly defining a quasi-distance metric \cite{BH14}.  Link scheduling problems can be formulated on edge-weighted
interference graphs \cite{hoeferspaa} that properly generalize both graph-based and SINR models, linked by a
graph-theoretic parameter.  Distributed dissemination problems, however, necessarily require metric restrictions, such
as doubling or \emph{fading} metrics \cite{H12}, and limits on communication abilities in order to capture both types of models.

The local broadcast and global broadcast problems have been extensively studied in both graph-based radio network models~\cite{ABLP89,ABLP91,CR03,DGKN13,DT10N,GHK13,GHLN12,KP03,KM98,N14} and the SINR model~\cite{BarenboimPeleg15,DGKN13,GMW08,HM11,JKRS14,SRS08,YHWL12,YHWTL12,YWHL11D}. For local broadcast, the best results known in the radio network model are both $O(\Delta\log n)$ with and without knowing (an upper bound on) $\Delta$~\cite{ABLP89,GHLN12}. In the SINR model, with knowledge of $\Delta$, the local broadcast can be accomplished in the same time bound as in the radio network model \cite{GMW08}. 
If $\Delta$ is not known, the best result is $O(\Delta\log n+\log^2n)$~\cite{HM11,YHWL12}, which is improved to $O(\Delta + \log n)$
with free acknowledgments \cite{HM11}. This can be further improved to $O(\Delta + \log n \cdot \log\log n)$ in the spontaneous case, when $\Delta$ is known~\cite{BarenboimPeleg15}. 

The time complexity of non-spontaneous broadcasting in the radio network model is $\Theta(D+\log n)\log(n/D)$~\cite{ABLP91,CR03,KP03,KM98,N14} without collision detection. With collision detection, this lower bound was recently broken in~\cite{GHK13}, where a solution of $O(D+polylog(n))$ was given. 
Broadcasting has also been treated in the SINR model under a variety of assumptions.
Some are stronger than ours (location information \cite{JKRS13,JKS13}, power control \cite{YHWTL12}),
while others relax the assumption about the connectivity property, incurring necessarily much higher complexity \cite{CKV15,DGKN13,JKS13I}. 
Results in our setting, but without carrier sensing, include time complexity of $O(D\log n \cdot polylog(R_s))$ \cite{DGKN13} (see also \cite{HHL15}), 
where $R_s$ denotes the maximum ratio between distances of stations connected in the communication graph;
and $O(D\log^2n)$ \cite{JKRS14}.
In the spontaneous setting, where the nodes can build an overlay structure along which the message is then propagated, Scheideler et al.~\cite{SRS08} 
used carrier sense to give a dominator algorithms, which can be applied to solve broadcast in $O(D+\log n)$ rounds. Yu et al.~\cite{YHWTL12} solved the problem in the same time bound using power control, while the algorithm in~\cite{JKRS14} that requires neither power control nor carrier sense runs in time $O(D\log n+\log^2 n)$. 

These problems have also been treated in dynamic networks. In the \emph{unstructured} model \cite{KMW04}, where nodes may wake up asynchronously (modeling the node insertion), the local broadcast problem is well studied, even in the SINR setting~\cite{GMW08,HM11,YHWL12,YWHL11D}, but this model does not consider node deletion. In the \emph{dual graph} model~\cite{FLN09} (originally due to \cite{CMS04}), both the local broadcast~\cite{GHLN12} and global broadcast~\cite{GKLN14,GLN13,FLN09,KLNOR10} problems are studied. But this model involves only edge behavior and not node changes (churn). Hence, the impact of dynamicity on wireless information dissemination is still largely unexplored. 

More detailed related work is introduced in Sec.~\ref{DRW}.

\mypar{Our Contributions}
%
We have obtained generalized and improved algorithms for two of the most fundamental dissemination problems, in some cases improving the best results known in static settings.
Beyond these specific results, we identify the following technical contributions:
\begin{compactenum}
\item \emph{Unified model of wireless networks}. 
 The model proposed appears to be the first that allows for the development of pan-model distributed dissemination algorithms.
This hopefully prompts further studies crossing the artificial boundary between graph- and fading-based models.
\item \emph{Dynamic networks under interference}. This appears to be the first work to address dynamic networks in the presence of comprehensive interference. 
 \item \emph{Uniform algorithms}. Our algorithms in the static spontaneous setting appear to be the first in fading models that work independent of instance parameters (number of nodes, max.~degree). 
\item \emph{Primitives for carrier-sense.} We introduce several primitives or \emph{capabilities} that can be implemented using environmental sensing,
and propose to study the power of such primitives.
\item \emph{Stabilization mechanisms.} We identify contention adaptation as a fundamental ability in wireless networks, that appears to be of crucial value to implement other distributed tasks.
\end{compactenum}
\mypar{Roadmap} The formal model and basic definitions, including the definitions of communication model and carrier sensing primitives,
 are given in Sec.~\ref{sec:model}. Sec.~\ref{sec:balancingcontention} contains the core technical part of the paper that includes the basic contention balance routine and its analysis. The main results concerning local and global broadcast problems are presented in Sections~\ref{sec:localbroadcast} and~\ref{sec:broadcast}, respectively. Due to space constraints, most proofs are relegated to appendices.

\section{Models and Definitions}\label{sec:model}

We consider a dynamic network of point-size wireless devices (nodes). Nodes can transmit messages in time slots/rounds  that are sufficiently long to allow a transmission of a single message. No global clock or synchronization of rounds is required, but the clocks of different nodes run at a similar rate, i.e., the length of a round differs between nodes at most by a factor of 2. Nodes may arrive and leave the network at any time. Unless specified otherwise, the nodes are assumed to work \emph{non-spontaneously}: they can initially be in sleep state and join the execution of an algorithm only after receiving a message. We say a node is \emph{alive} at some point in time if it is present in the network. We assume the total number of nodes in the network is polynomially bounded by a number $n$ in each round.
 We use $V$\marginpar{$V$} to denote the set of alive nodes  at any fixed point in time and also use $n$ to denote the current number of nodes, i.e. $n=|V|$.\marginpar{$n$}

We assume  all nodes use the same transmission power $P$\marginpar{$P$} for communication in all rounds. 

\thead{Metrics}
The \emph{signal strength} -- or \emph{interference}, depending on context -- 
of transmitting node $u$ on a node $v$ is $\cI_{uv}=P/f(u,v)$\marginpar{$\cI_{uv}$}, where $f(u,v)> 0$\marginpar{$f(u,v)$} is the \emph{path loss} from $u$ to $v$. 
%
The \emph{metricity} of a space $(V,f)$ is the smallest number $\z$\marginpar{$\z$} such that for every triplet $u,v,w\in V$,
$f(u,v)^{1/\z}\leq f(u,w)^{1/\z}+f(w,v)^{1/\z}$ \cite{BH14}.
%
%
We define $d(u,v)=f(u,v)^{1/\z}$\marginpar{$d(u,v)$} if $u$ and $v$ are different nodes and $d(u,v)=0$ when $u = v$. Note that $(V,d)$ is a \emph{quasi-metric}, as all metric axioms except symmetry hold.
In the rest of the paper, we  assume that in each round the metricity of the network is bounded by a fixed constant $\z$ and will work with values $d(u,v)$ instead of $f(u,v)$.
We assume  the quasi-metric $(V,d)$ has \emph{bounded independence}, defined below,  roughly stating that there cannot be many nodes each causing high interference to a \emph{fixed} node, while having low mutual interferences. 

First, some notations. The \emph{ball} with radius $r$ centered at $u$ is defined as\marginpar{$B(u,r)$} 
$B(u,r)=\{v\in V|\max\{d(v,u),d(u,v)\}<r\}$.
The \emph{in-ball} with radius $r$ centered at $u$ is defined as $D(u,r) = \{v \in V | d(v,u) < r\}$;\marginpar{$D(u,r)$}
 clearly, $B(u,r)\subseteq D(u,r)$.
A set $S\subseteq V$ is a $r$-\emph{packing} for set $S'$ if balls of radius $r$ centered at nodes in $S$ are contained in $S'$ and are disjoint. $S$ is a $r'$-\emph{cover} for $S'$ if the union of balls of radius $r'$ centered at nodes in $S$ contains $S'$.
Note that any maximal $r$-packing is a $2r$-cover, and thus one can bound sizes of covers by packings.

We say that $(V,d)$ has $(\rmin,\l)$-\emph{bounded independence}\marginpar{$\rmin$}\marginpar{$\l$}, for given $\rmin\ge 0$ and $\l>0$, if for every $q\ge 1$ and every in-ball $D$ of radius $qr_{min}$, the size of a maximum cardinality $r_{\min}$-packing of $D$ is at most $C\cdot q^\l$, where $C$ is a constant, possibly depending on $\l$. 
%
For instance, the Euclidean plane is $(r,\l=2)$-bounded independent, for every $r$.

\thead{Neighborhoods, Communication Graph and Dissemination Problems}
Let $R$ denote the maximum transmission distance possible when no other node transmits. 
As the latter event is arguably very rare, we define the \emph{communication radius} $R_B = (1-\epsilon) R$\marginpar{$R_B$} as a slightly smaller distance, 
where $\e$\marginpar{$\e$} is a \emph{precision} parameter. We will drop the parameter $\e$ whenever it is fixed and clear from the context.
Fix a round $t$. 
The \emph{neighborhood} of a node is $N_t(u, \e) = \{v\in V: d(u,v) \le (1-\e)\R\}$,\marginpar{$N_t(u)$} describing who $u$ can communicate with directly.
The basic operation of interest is when $u$ broadcasts a message to its neighbors $N_t(u)$.
The \emph{communication graph} is a directed graph $G_t(V,E)$,\marginpar{$G$, $G_t$} where $(u,v)\in E$ if and only if $v\in N_t(u)$. Thus, the sequence $G_0,G_1,\dots$ defines a \emph{dynamic graph}.
%
The \emph{vicinity} of $u$ refers to a larger region, $D_u^\rho=D(u,\rho\R)$\marginpar{$D_u^\rho$}, for a parameter $\rho > 1$. 


The data dissemination problems that we consider are defined below. We say that a node $u$ \emph{mass-delivers} in round $t$ if it transmits and all its neighbors ($N_t(u)$) receive the message.
\begin{compactitem}
\item{In the \emph{Local Broadcast} problem, given a node $u$, it is required to minimize the time from the beginning of the algorithm until node $u$ mass-delivers at least once, assuming it stays alive during that time.}

\item{In the \emph{(global) Broadcast} problem, given a distinguished \emph{source} node that initially holds a message, the goal is to minimize the time needed to deliver the message to every node in the network through multihop transmissions.} 
\end{compactitem}

\thead{One Hop Communication}
%
When is a transmission successfully received?
Suppose a node $u$ transmits in round $t$, and let $S$ be the set of concurrently transmitting nodes.
Let $\rho_c=\rho_c(\e)\ge 0$\marginpar{$\rho_c$} and $\Ic=\Ic(\e)>0$\marginpar{$\Ic$} be parameters that depend on the precision $\e$.
\begin{definition}
(\scc, Success on a clear channel) If no other node in $D_u^{\rho_c}$ transmits \emph{and}  the total interference at node $u$ is at most $\Ic$ (i.e. $S\cap D_u^{\rho_c}=\emptyset$ and $\sum_{v\in S}\cI_{vu} \le \Ic$), then the transmission of node $u$ is successfully received by \emph{all} its neighbors ($N_t(u)$).
Otherwise, the reception is under adversarial control.
\end{definition}

\thead{Randomized Algorithms}
%
We mainly consider randomized algorithms of the following form: in each round $t$, node $v$ makes a transmission with probability $p_t(v)$, independent of other nodes' transmissions in that round.
An important notion for the analysis of such algorithms is \emph{local contention}, the sum of the transmission probabilities in a close region.
%
The contention in the \emph{close neighborhood} of a node $v$ in round $t$ is \marginpar{$P_t(v)$}
$P_t(v)=\sum_{w\in B(v,\R/2)}p_t(w)$,
where the radius $\R/2$ allows all pair of nodes in $B(v,\R/2)$ to potentially communicate.
Also, let
$P_{t}^{\rho}(v)=\sum_{u\in D_v^\rho}p_t(u)$\marginpar{$P_t^\rho(v)$}
denote the contention in the larger \emph{vicinity} of $v$ in round $t$ ($\rho$ will be fixed later).
We will also use the notation $\cI_t^\rho(v)$\marginpar{$\cI_t^\rho(v)$} to denote the interference at $v$  from nodes outside its vicinity (in $\bar D_{v}^\rho=V\setminus D_v^\rho$) in round $t$. The expected value of $\cI_t^\rho(v)$ is\marginpar{$\bcI_t^\rho(v)$} $\bcI_t^\rho(v)=\sum_{w\in \bar D_{v}^\rho}p_t(w)\cI_{wv}$.

\thead{Sensing Primitives}
We assume the nodes have abilities to sense activity on the channel.
Namely, we assume the nodes are able to detect high and low contention in their vicinity, detect (under some conditions) whether their transmission in a given round succeeded and detect a single very near transmission. In the following, we formalize these notions in three primitives: {\cd}, {\ack} and {\ntd}. We show in Sec.~\ref{sec:primitiveproofs} how 
\emph{all} these primitives can be implemented with \emph{basic} physical carrier sensing and possibly also with other means.

\ul{Contention Detection} ({\cd}).
Contention can be probabilistically deduced from measured level of radio activity. We want to relax this ability and will use 
the following variant, where the outcome of {\cd} is one of the two values: {\busy} or {\idle} channel.
Formally, 
for each node $v$ and round $t$:
\begin{compactitem}
\item{if contention among close neighbors is high ($P_t(v)>\phi$) then they \emph{all} detect {\busy} channel in round $t$ with probability at least $1-h_1^{-\phi}$, for given $\phi > 1$, where $h_1>1$ is a constant\marginpar{$h_1$},}

\item{if 
the contention in the vicinity of $v$ is low ($P_{t}^\rho(v)\le \eta$) \emph{and} the interference on $v$ from outside its vicinity is above a threshold ($\cI_t^\rho(v)< \Icd$) then $v$ detects {\idle} channel in round $t$ with probability at least $h_2^{-\eta}$\marginpar{$h_2$}, for given $\eta > 0$,  where $\Icd>0$\marginpar{$\Icd$} and $h_2 > 1$ are constants.}
\end{compactitem}

\ul{Successful Transmission Detection} (\ack).
If a node $u$ has the {\ack}={\ack}($\e$) primitive (depending on the precision parameter $\e$) then: if $u$ transmits in round $t$, the interference at $u$ is bounded by $\Iack$\marginpar{$\Iack$} and the transmission is received by all nodes in $N_t(u, \e)$, then the outcome of {\ack} is 1, where $\Iack$ is a parameter. If the transmission is not received by a node $v\in N_t(u, \e)$ then the outcome is $0$. Otherwise, the outcome is $0$ or $1$, adversarially.

\ul{Near Transmission Detection} (\ntd).
With {\ntd}={\ntd}($\e$) primitive, a node is able to detect if a transmitter is very close, assuming that it receives the transmitted message. The outcome of {\ntd} is $1$ for node $v$ in round $t$ if $v$ \emph{receives} a transmission from a node $u$, $u\in D_v^{\e/2}$. Otherwise, the outcome of {\ntd} is $0$.
This can also be made approximate.

\thead{Dynamicity} We consider a dynamic network where the topology may change adversarially in each round due to node churn (node arrivals/departures) and edge changes. 
We assume that arriving nodes start running the algorithms from an initial configuration, so we do not limit the rate of churn. With edge changes, existing nodes that were not neighbors before, may become neighbors (e.g. due to mobility). The new neighbors may cause too much interference in a too short time, so the edge changes should be limited. We assume the amount of edge changes is bounded for each node $v$, as follows. Consider a time interval $T$ of length $\Omega(\log n)$. We require that the number of new neighbors of $v$ during $T$ (not counting churn) is bounded by $\tau |T|$, where $|T|$ denotes the number of rounds in $T$ and $\tau$\marginpar{$\tau$} is a constant, to be fixed later. We further assume   the fraction of rounds in $T$ when there are more than $\phi$ new neighbors of $v$ is bounded by $O(\phi^{-k})$ for every $\phi\ge 1$, where $k > 2\l/(\z - \l)$.\marginpar{$k$} Note that there is no restriction on distance changes inside the neighborhood of $v$ (e.g. it is fine for node $v$ if its neighbors move, as far as they remain neighbors). Note that the edge changes may affect the underlying metric, but we require that the upper bounds on metricity and independence are maintained.




\thead{Requirements and Assumptions}
For the convenience of the reader, we gather together all of our assumptions and requirements in a single place.

Communication is assumed to succeed in a clear channel (\scc).
We assume constant metricity $\zeta$ and that $(V,d)$ has $(\rmin,\l)$-bounded independence with $1\le \l<\z$. For the local broadcast problem, we assume that $\rmin\le\R/4$, and for the broadcast problem that $\rmin \le \e\R/4$. 
As is standard in fading models, the communication radius $R_B$ is necessarily an $(1-\epsilon)$-fraction of the maximum transmission distance in a clear channel. 

Besides the knowledge required by the primitives that are needed for a particular algorithm, the nodes are assumed to know the precision parameter $\e$. A polynomial estimate on the number of nodes, $n$, is needed  
in dynamic and non-spontaneous algorithms, but not in the static spontaneous problems.
Knowledge of approximations of model parameters are needed to implement primitives, including $\epsilon$, $\zeta$, $R$, $\rho_c$, and $\Ic$.
Knowledge of the maximum degree $\Delta$ is not needed.


%
Synchronous operation is only assumed in the Broadcast algorithm.
Aspects not defined or constrained are assumed to be under (adaptive) adversarial control, including
when transmissions that fail {\scc} are successful, or when nodes appear or disappear from the network.

The extent of increases in edge strengths over a period is restricted, as detailed above, while decreases are not and neither are node changes.
The dissemination problems are only expected to function with the set of nodes that are sufficiently stable, as detailed in the respective section.


\section{Controlling Contention}\label{sec:balancingcontention}

In order to keep the contention in the network balanced, we propose a basic procedure called {\ta}, which will be the main building block in our algorithms. 
The idea is to let each node adapt its transmission probability to the contention detected using the assumed {\cd} primitive.
The parameter $\beta\ge 1$ describes the passiveness of the newly arriving nodes. 

\begin{quote}
\ta($\beta$): Each node $v$ maintains transmission probability $p_t(v)\le 1/2$ in each round $t$,
initialized as $p_t(v)=\frac{1}{2}n^{-\beta}$ when $v$ enters the network. 
In round $t+1$, $v$ does:\\
1. {Transmit with probability $p_t(v)$, and}\\
2. {Set $p_{t+1}(v)= \begin{cases}
    \max\{p_t(v)/2,n^{-\beta}\} & \mbox{if {\busy} channel,  and} \\
       \min\{2p_t(v), 1/2\}, & \mbox{otherwise.} 
       \end{cases}$}
\end{quote}

The aim for controlling contention is, of course, to ensure that transmissions made have a fair chance of being successful, which means they sufficiently overpower the interference experienced at intended receiver from all other transmissions made in that round.
We account for this interference in two ways: the \emph{contention} captures the expected interference from the nodes' neighbors, while the \emph{interference} integrates also the interference from nodes further away. 

We will measure the contention in the \emph{vicinity} of each node $v$, i.e. in $D_v^\rho$, where $\rho$\marginpar{$\rho$} is a large enough constant. We specify a threshold $\eta=\log_{h_2} (10/9)$\marginpar{$\eta$} (recall $h_2$ from {\cd} definition) for measuring contention: if $P_t^\rho(v) > \eta$ then round $t$ is a \emph{high contention round} for node $v$ and is a \emph{low contention round}, otherwise. We further specify a threshold $\hcI$ for interference: if $\hcI_t^\rho(v)>\hcI$ then round $t$ is a \emph{high interference round} for $v$ and is \emph{low interference round}, otherwise. 

These thresholds are chosen so as to ensure that in a low-contention/interference round, node $v$ will be likely to succeed if it transmits. However, requiring all or most rounds to be low contention for all nodes will lead to high delays. Instead, it turns out that most rounds will have \emph{bounded contention} and low interference, which 
allows for good progress; we say node $v$ experiences \emph{bounded contention} in round $t$ if $P_t^\rho(v) < \hphi$\marginpar{$\hphi$}, where $\hphi>0$ is a large enough constant, to be specified later.
We say that a round $t$ is \emph{good} for node $v$ if $t$ is both bounded contention and low interference round
(in which case, some node in $v$'s vicinity has a good chance of successfully transmitting).

We analyze the properties of {\ta} using the notion of \emph{a phase}, the shortest time in which at least $\gamma\log n$ rounds occur for all nodes,
where $\gamma$ is sufficiently large (given in Prop.~\ref{pr:Pgoodroundl}). 
We use $H$ to denote a general phase and also the set of rounds in that phase. $|H|$ denotes the number of rounds in a phase $H$ (i.e. $\gamma \log n$).
%
The fundamental property of {\ta} (Prop.~\ref{pr:Pgoodroundl})
is that, for each node and each phase, most of the rounds in the phase are good for that node.
%
This property is then used to show that:
\begin{compactenum}
\item{If most of the good rounds in a phase have low contention, then node $v$ detects {\idle} channel in most of those rounds.}
\item{Otherwise, during a constant fraction of the rounds, a node in the vicinity of $v$ mass-delivers.}
\end{compactenum}

By choosing the parameters carefully, i.e.~requiring \emph{low enough} contention, we can make sure that during a phase with mostly low contention, the node detects {\idle} channel in \emph{most of the rounds} (more than half) in a phase, thus leading to an increase of transmission probability by the end of the phase, which, after sufficiently many phases ensures message delivery, w.h.p. On the other hand, during a phase with mostly high contention, there will be many nodes in the vicinity of $v$ that successfully transmit, leading to lower contention.
These ideas are applied in Thms.~\ref{thm:localbroadcast} and~\ref{thm:broadcast}. The core idea behind this analysis is based on~\cite{SRS08}.

%
%

{\begin{proposition}\label{pr:Pgoodroundl}
Let $\sigma\in (0,1)$.
If constants $\rho=\rho(\sigma,\hcI),\hphi = \hphi(\rho,\sigma)$ and $\gamma=\gamma(\rho,\hphi,\sigma,\hcI)$ are large enough then for each node $v$ and phase $H$, with probability $1-O(n^{-3})$, 
a $(1-\sigma)$-fraction of the rounds in $H$ are good.
\end{proposition}}

The proof is rather technical and is deferred to Sec.~\ref{sec:contentiondetails} but the intuition is as follows.
The contention in each neighborhood must be bounded most of the time, because when it becomes large, it has a high chance of being decreased due to {\busy} channel.
Moreover, we show that in expectation, the contention in \emph{all} local neighborhoods is bounded, which is then combined with a geometric argument to show that the expected interference at each node is low most of the time.
\smallskip

We derive from the fundamental property two useful propositions.
The first says that if contention is high, then nodes in the vicinity deliver the message.

\begin{proposition}\label{pr:Pdecrease}
Assume that constants $\hphi, \rho,\gamma$ are large enough. 
For each node $v$ and phase $H$, 
if at least $1/10$-fraction of the rounds of $H$ are of high contention,
then $\Omega(|H|)$ nodes in $D_v^{\rho}$ mass-deliver, with probability $1-O(n^{-3})$.
\end{proposition}
To this end, we first show that if a round is good for node $u$ and a node in its vicinity transmits, then it mass-delivers with constant probability,
utilizing both metric assumptions and the properties of good rounds.
We then argue that since most rounds are good (by Prop.~\ref{pr:Pgoodroundl}) and most rounds have by assumption sufficient contention, 
many rounds will be both good and with sufficient contention, and in each of those, a node in the vicinity of $v$ is likely to transmit and succeed.

\smallskip
When contention is low in a lot of rounds of a phase, 
the node will  detect {\idle} channel by the {\cd} primitive in many rounds.
This will actually happen during many good rounds, which have the low local contention and low external interference to allow for this detection.
\begin{proposition}\label{pr:idle}
Assume that $\hphi, \rho,\gamma$ are large enough. 
For each node $v$ and phase $H$, if at least $9/10$-fraction of the rounds of $H$ are low contention rounds, 
then  with probability $1-O(n^{-3})$, in at least $3/5$-fraction of the rounds of $H$, $v$ will detect {\idle} channel and have low contention and low interference.
\end{proposition}

\section{Local Broadcast}\label{sec:localbroadcast}

We propose an algorithm for asynchronous local broadcast in dynamic networks. The algorithm is an extension of the {\ta} procedure, where the nodes try to balance the contention in the network and stop transmitting as soon as they deliver their message. We assume the nodes are powered with {\cd} and {\ack} primitives.
Note that the passiveness parameter is set to $\beta=1$, which means that the transmission probability of nodes does not get below $1/(2n)$. 
\begin{figure}[h!]
\lbr: Each node $v$ executes {\ta}(1) with the following additional step:
if $v$ transmits and detects {\ack}, it stops (i.e. $p_{r}(v)=0$ for $r>t$).
\end{figure}

We will estimate the performance of the algorithm using the notion of \emph{dynamic degree}, defined as follows. Given a parameter $\rho>0$,
we denote $\Delta_v^\rho(t,t')=|\cup_{r=t}^{t'}D_v^\rho(r)|$ for node $v$ and rounds $t, t'$ with $t' > t$, where $D_v^\rho(r)$ denotes the in-ball $D_v^\rho$ in round $r$.

Below we prove that if there are not too many node insertions in the neighborhood of a node $v$, then $v$ mass-delivers (delivers to all its neighbors) in time comparable to its dynamic degree with $\rho$ a constant.
The main tools for proving the bound are Props.~\ref{pr:Pdecrease} and~\ref{pr:idle}. 
First we argue that if there is a phase of mostly low contention, then node $v$ will deliver its message, w.h.p. Then we show that if the insertions are not too intensive then the contention around $v$ will decrease and a phase with mostly low contention will happen.
\begin{theorem}\label{thm:localbroadcast}
There is a constant $\rho>0$, such that a node $v$ performing {\lbr} asynchronously in a time interval $T=[t,t']$ with $t'-t=\Omega(\Delta_v^\rho(t,t') + \log n)$ mass-delivers, w.h.p., provided that $t'-t=O(n^2)$.
\end{theorem}
Note that the assumption $t'-t=O(n^2)$ is needed only for making the claim w.h.p.: it can be relaxed to higher degree polynomials by only increasing constant factors.
\begin{proof}
Let us fix constants $\rho,\hphi,\gamma$ so that Props.~\ref{pr:Pdecrease} and~\ref{pr:idle} with $\hcI=\min\{(1-1/\rho)^\z\Ic,\Icd,\Iack\}/10$.

 We partition $T$ into phases (for node $v$) and classify them into two types: (type A) phases $H$ where at least $1/10$-fraction of rounds are high contention rounds (i.e. $P_t^\rho(v) \ge \eta$), and (type B) phases $H$ where at least $9/10$-fraction of rounds are low contention rounds (i.e. $P_t^\rho(v) < \eta$).

\begin{claim}
Node $v$ mass-delivers in a type B phase, w.h.p.
\end{claim}
\begin{proof}
Let $H'$ be the low contention and interference rounds during phase $H$ where $v$ detects {\idle} channel. By Prop.~\ref{pr:idle} $Pr[|H'| \ge 3|H|/5] > 1-O(n^{-3})$. Assume for now that the latter happens. For each $t\in H'$ we have $p_{t+1}(v)=\min\{2p_t(v), 1/2\}$. Let us call this operation doubling.
The value of $p_t(v)$ at the beginning of the phase is at least $1/(2n)$, so $\log n$ doubling operations are sufficient to raise it to $1/2$. The probability can be further halved during the phase at most $|H\setminus H'|\le 2|H|/5$ times. Thus, we may assume  we have at most $2|H|/5 + \log n$ halving and at least $3|H|/5$ doubling operations applied to an initial value $1/2$. If  $\gamma > 10$, then the total number of halving operations is less than $|H|/2$. It follows that $v$ has $p_t(v)=1/2$ in at least $(3/5-1/2)|H|=|H|/10$ low contention/interference rounds. By Lemma~\ref{le:broadgood}, in each such round, $v$ mass-delivers with probability at least $0.9\cdot \frac{1}{2}\cdot 4^{-\eta}$; hence, if $\gamma$ is large enough, $v$ mass-delivers in $H$, w.h.p.
\end{proof}

It remains to argue that there will be a type B phase during time interval $T$.
Consider a type A phase $H$.
Prop.~\ref{pr:Pdecrease}  implies that with probability $1-O(n^{-3})$, $\Omega(\log{n})$ nodes in $D_v^\rho$ deliver their message and stop during phase $H$.
 Thus, with probability $1-O(n^{-1})$, there are at most $C\cdot\frac{\Delta_v^\rho(t,t')}{\log n}$ type A phases with $C>0$ a constant, as there are at most $\Delta_v^\rho(t,t')$ nodes in  $D_v^{\rho}$ during the time interval $T$ (also recall that $t' - t =O(n^2)$). We conclude that if $T$ consists of at least $C\cdot\frac{\Delta_v^\rho(t,t')}{\log n} + 1$ phases, it will contain a type B phase and $v$ will deliver its message w.h.p.
 \end{proof}

\mypar{Implications for Static Networks}
In static networks, the parameter $\Delta_v^\rho(t,t')$ is at most $|D_v^\rho(0)|=O(\Delta)$ if $\rho$ is constant, where $\Delta=\max_{v}\{|N(v)|\}$ is the maximum size of a neighborhood in the network. Thus, we obtain the following optimal result (up to constant factors) for static networks, as $\Delta$  and $\log n$ are lower bounds even when  running in the spontaneous mode~\cite{YHWL12}.
\begin{corollary}
When running {\lbr} in a static asynchronous network, each node $v$ completes local broadcast in $O(|D_v^\rho| + \log n) = O(\Delta + \log n)$ rounds, w.h.p.
\end{corollary}

\textbf{Remark.} In the special case when the nodes can start executing the algorithm simultaneously, i.e. in the \emph{spontaneous mode}, the nodes need not know an upper bound on the size of the network. Indeed, each node may start running {\ta} with initial probability set to an arbitrary value and with no lower limit. The first phase will be spent for stabilization and can be ignored, while the argument for the rest of the phases is nearly identical to the one in Thm.~\ref{thm:localbroadcast}.

\section{Broadcast}\label{sec:broadcast}

For the broadcast problem, we assume  nodes communicate in synchronized rounds of equal length.  Each round consists of two slots. The idea is to use   the first slot of each round for disseminating the message with {\ta} and the second slot for  notifying nodes which have no uninformed neighbors. The latter is accomplished by using higher precision primitives,  namely {\ack}($\e/2$) and {\scc}($\e/2$), when executing {\ta}. This helps to detect a transmission that is successfully received by all nodes in $N(v,\e/2)$ of a node $v$. Upon detecting such a transmission, node $v$ resends the message in the second slot, in order to inform nodes $u$ with $v\in D_u^{\e/2}$ that their neighborhood ($N(u,\e)$) has been covered. A node $u$ can detect that $v\in D_u^{\e/2}$ using {\ntd}.

The algorithm is presented below. We assume  the passiveness parameter $\beta$ of {\ta} is large enough, to be defined later. Note that the algorithm works for the non-spontaneous mode, as nodes act only after receiving the message.

\medbreak

\br($\beta$): Initially, only a source node $s$ has the message. A node $v$, upon receiving a message, starts executing {\ta}($\beta$) in the first slot of rounds. In addition, in each round $t$,\\
1. if $v$ detects {\ack} in the first slot, it retransmits in the second slot and restarts {\ta}($\beta$),\\
2. if $v$ receives a message in the first slot and detects {\ntd} in the second slot, it  restarts {\ta}($\beta$).
\smallskip

In order to evaluate the progress of the algorithm, we use a notion of a \emph{dynamic distance}, as defined below. Let $c>0$ be a parameter.
A sequence $v_1=s,v_2,\dots,v_k=v$ is called a \emph{stable $s$-$v$ path} if there is a sequence $I_1,I_2,\dots,I_{k-1}$ of time intervals with $I_i=[b_i, e_i]$, such that $e_i-b_i\geq c\log n$, $e_i-e_{i-1}\geq c\log n$ and nodes $v_{i-1}$ and $v_i$ are both alive and $v_i\in N(v_{i-1},\e)$ during $I_i$. The \emph{time-length} of a stable $s$-$v$ path is $e_{k-1} - b_1$. The \emph{stable $s$-$v$ distance} $D_{st}^c(s,v)$\marginpar{$D_{t}^c(s,v)$} is defined as the minimum time-length of a stable $s$-$v$ path.
Note that a stable path need not be connected at any fixed point in time. Moreover, most of the nodes might be missing at any given point in time.


The core idea behind the analysis of the following theorem is similar to the case of local broadcast:
we show that as soon as a neighbor $u$ of a node $v$ has the message and $u$ and $v$ keep being neighbors for $O(\log n)$ rounds, $v$ will receive a transmission of $u$ during those rounds. 


\begin{theorem}\label{thm:broadcast}
Assume the edge change rate $\tau$ is  sufficiently small. There are constants $\beta,c>0$, such that when running {\br}($\beta$) in the synchronous mode, each node $v$ receives the message in $O(D_{st}^c(s,v))$ rounds w.h.p.
\end{theorem}

\mypar{Implications for Static Networks}
When the network is static, Theorem~\ref{thm:broadcast} can be reformulated in terms of hop-distance $dist_G(s,v)$ in the communication graph, which is defined as the length of the shortest directed $s$-$v$ path in $G$: we have that $D_{st}^c(s,v)=O(\log n)\cdot dist_G(s,v)$ for any node $v$. Note also that in this setting nodes that succeeded transmitting or detected {\ntd} need not continue the algorithm, so they stop transmitting. In this case, setting the passiveness parameter to $\beta=1$ suffices. We call this variant of the algorithm {\br}$^*$.

\begin{corollary}
When running {\br}$^*$ in synchronous non-spontaneous mode in a static network with source $s$, each node $v$ receives the message in $O(\log n)\cdot dist_G(s,v)$ rounds w.h.p. When the communication graph is strongly connected, the broadcast from any source node is completed in $O(\log n)\cdot D_G$ rounds, where $D_G$ is the diameter of the communication graph.
\end{corollary}

In the spontaneous mode, the bound above can be further improved to $O(D_G+\log n)$; see Appendix \ref{sec:spontaneous}.
This is based on finding a constant-density dominating set in $O(\log n)$ time \cite{SRS08} and simultaneously propagating along the dominators in $O(D_G+\log n)$ time.
We can therefore extend the approach based on \cite{SRS08} to uniform algorithms in bounded-independence metrics.

These results are close to best possible. We show below that in order to obtain bounds of that magnitude, it is necessary to have the {\ntd} primitive.  To this end, we extend the lower bound construction of~\cite[Thm.\ 7]{DGKN13} for ``compact SINR'' to our setting.

This construction leverages the property of our model that there can be arbitrarily many nodes that are mutually close to each other.
Namely, the bounded-independence metric is strictly more relaxed than the standard Euclidean metrics.
Indeed, there is a $O(D\log^2 n)$-round broadcast algorithm for the SINR model that does not need {\ntd} or other carrier sensing primitives \cite{JKRS14}. What the lower bound then illustrates is that to obtain such results, one must depend on opportune traits of the SINR model that we have tried to avoid and are not necessary for problems like local broadcast. Thus we can observe concrete tradeoffs depending on model assumptions.

\begin{theorem}\label{thm:brlowerbound}
 For every (possibly randomized) broadcast algorithm $\mathcal{A}$ that uses neither node coordinates nor {\ntd} primitive, there is a $(\e\R/8, 1)$-bounded-independence metric space where $\mathcal{A}$ needs $\Omega(n)$ rounds to do broadcast in a $O(1)$-broadcastable network, even if the nodes have {\cd} and {\ack} primitives and operate spontaneously.
\end{theorem}

\newpage


\appendix
\section{Other Related Work}\label{DRW}


\emph{Wireless models:} 
Considering wireless interference, there are two classes of wireless network models: graph-based and physical models. Basically, the graph-based models define a local and binary type of interference, while the physical models consider fading effect of signal on wireless channels. The most classical graph-based model is the radio network model~\cite{CK85}. In this model, the network is modeled using a communication graph, where each pair of nodes that can communicate with each other is connected by an edge. It defines the interference just from direct neighbors, and a transmission can succeed if and only if there is only one neighbor of the receiver transmitting.  There are many widely used variants of the classical radio network models, including: 1) the $k$-hop model where the interference comes from $k$-hop neighbors~\cite{SW06}; 2) Unit Disc Graph (UDG) model~\cite{CCJ90} which defines the neighborhood using a unit disc; 3) Quasi Unit Disc Graph (QUDG) model~\cite{KWZ08} which just defines all pairs of nodes with distance at most $\rho$ for some given $\rho\in(0,1]$ are adjacent, and leave the `grey' area in $(\rho, 1]$ being determined by an adversary; 4) Protocol model~\cite{GK00}, where each node has a transmission range and an interference range, and a successful transmission occurs if a node falls into the transmission range of a transmitter and outside the interference ranges of all other transmitter; 5) Bounded-Independence Graph (BIG) model~\cite{SW10}, which defines abstractly and requires that the size of the maximal independent set in the $r$-hop neighborhood of each node is bounded by a polynomial function with $r$. Though the graph-based interference models miss certain crucial aspects of actual wireless networks, the simple definition of these models can help derive novel insights into distributed solutions to wireless problems. 

Physical models, also known as SINR models~\cite{GK00}, capture the fading and cumulative features of receptions in actual wireless environments.
The default assumption is that interference fades with a polynomial of the distance, and transmission succeeds only if the received signal strength is sufficiently larger than the total interference plus noise. Recently, the SINR model has attracted great attentions in the distributed community~\cite{BHM13P,BH14,DGKN13,GMW08,HHMW13,HM11i,HM12P,HWY15,JKRS14,JKS13,KV10,SRS08,YHWL12,YHWYL13,YWHL11D,YWYYL15}. Most of these works focus on networks embedded in Euclidean space, while many of the results hold also for doubling or ``bounded growth'' metrics \cite{BH14,DGKN13,JKRS14}. 
Those doubling metrics constrain growth at every (or arbitrarily small) granularity, 
while ours only bounds regions proportional to the transmission range, in particular capturing bounded independence graphs (BIG).
For more on wireless models, please refer to~\cite{LW12,SW06} .

\emph{Local Broadcast}:
In the radio network model, probably the first local broadcast result was a
randomized algorithm of Alon et al.~\cite{ABLP89} in a synchronous model, running in 
$O(\Delta\log n)$ rounds. 
Derbel and Talbi~\cite{DT10N} later generalized their algorithm to work without knowledge of $\Delta$ and their proposed algorithm can accomplish local broadcast in $O(\Delta\log n+\log^2n)$ rounds. The decay strategy also yields an $O(\Delta\log n)$ time algorithm for local broadcast without knowledge of $\Delta$~\cite{BGI87,GHLN12}. 

Goussevskaia et al.\ \cite{GMW08} gave the first results for local broadcast in the SINR model, running in time
$O(\Delta\log n)$ and $O(\Delta\log^3n)$ with and without knowledge of $\Delta$, respectively.  The latter was improved in
\cite{YWHL11D} and further improved, independently, to $O(\Delta\log n+\log^2n)$ time \cite{HM11,YHWL12}. 
With free acknowledgements, this was improved to $O(\Delta + \log^2 n)$ \cite{HM11}.
When additionally $\Delta$ is known, this was further improved recently to $O(\Delta + \log n \cdot \log\log n)$ in the spontaneous setting
\cite{BarenboimPeleg15}. The speedup of multiple channels on local broadcast was considered in~\cite{HWY15,YWYYL15}

\emph{Broadcast}:
The complexity of broadcasting is well understood in graph-based models. In the radio network model, Bar-Yehuda et al.~\cite{BGI87} presented the decay protocol which can accomplish non-spontaneous broadcast in $O(D\log n+\log^2n)$ rounds, where $D$ is the diameter. This result was improved to $\Theta(D+\log n)\log(n/D)$ independently by Czumaj and Rytter~\cite{CR03}, and Kowalski and Pelc~\cite{KP03}. These algorithms can be viewed as clever optimizations of the decay protocol and match the lower bound~\cite{ABLP91,KM98,N14}. With collision detection, this lower bound was recently broken in~\cite{GHK13}, where a solution of $O(D+polylog(n))$ was given. 
Broadcast in multi-channel radio networks was considered in~\cite{DGKN11,GGNT12}.
For the UDG model, an $O(D + \log^2 n)$ time algorithm was given in~\cite{DGKN13} in the spontaneous setting.


\emph{Distributed models of temporal variability}: 
Dynamic networks have been studied extensively in recent years (see \cite{KO11} for a survey),
but generally not in the presence of interference.

The \emph{dual graph} model~\cite{FLN09} (originally due to \cite{CMS04}) was designed to capture inherent \emph{unreliability} in wireless networks, much of which can be due to dynamicity.
The main focus of that work is on extending the radio network model in general graphs. Importantly, the dual graph model does not distinguish between interference and communication edges; only that the unreliable edges can transmit both interference and the usual communication, but their availability is under adversarial control.
Thus, there is no way to capture interference from further away nodes.
Most problems become extremely difficult against a powerful adversary, and to get good result, one must assume a much weaker one~\cite{GLN13}. This model only involves edge behavior and not node changes (churn). Both the local broadcast~\cite{GHLN12} and global broadcast~\cite{FLN09,KLNOR10,GLN13,GKLN14} problems are studied in the dual graph model.

A dynamic model that considers node insertion is the \emph{unstructured} model~\cite{KMW04}, which admits arbitrary wake-up mode and asynchronous communication.
This model was first proposed in the unit-disc setting, and then extended to bounded independence graphs (BIG)~\cite{SW09} and SINR~\cite{GMW08}. It has been widely used in the solution of a variety of distributed wireless problems~\cite{KMW04,MW08,SW09,SW10,HM11,GMW08,YWHL11D,YHWL12,YWYYL15}, including local broadcast~\cite{HM11,GMW08,YWHL11D,YHWL12,YWYYL15}, but there are no known global broadcast results. The model neither consider node deletion nor edge changes.
Hence, the impact of dynamicity  on global communication is largely unexplored. 


\section{Implementing Communication and Primitives}\label{sec:primitiveproofs}

\paragraph{Modeling communication} Our communication model captures most known algorithmic wireless models, as it is demonstrated below on the example of SINR, disk-graph based and Protocol models. 

Note that we assume below the distance $d(x,y)$ to be symmetric, but the results hold also for ``almost symmetric'' functions, i.e. when there is a constant $c$ such that $d(x,y) \le c\cdot d(y,x)$ for all $x,y$. 

\textit{SINR Model.} Consider a network in a metric space. In the SINR model of communication, if $S$ is the set of simultaneously transmitting nodes in the network, a node $v$ receives the transmission of node $u$ if and only if
\[
\frac{P}{d(u,v)^\z} > \beta \cdot \left(\sum_{w\in S\setminus \{u,v\}}\cI_{wv} + N\right),
\]
where constants $\beta\ge 1$ and $N> 0$ denote the minimum SINR threshold and the ambient noise, respectively.

Note that $\R=(P/(\beta N))^{1/\z}$ in this setting.
We can implement {\scc} with parameters $\Ic=\min\{\beta, ((1-\e)^{-\z}-1)\}N/2^{\z}$ and $\rho_c=0$, as shown in the proposition below.

\begin{proposition}\label{pr:scc}
If the interference at a node $v$ is less than $\Ic$ in round $t$, then $v$ will deliver its message if it transmits.
\end{proposition}
\begin{proof}
Note that if the interference at node $v$ is not more than $\Ic$ then there is no node $u\in D_v^{2}$ transmitting in round $t$, as otherwise the interference at $v$ would be at least $P /(2\R)^\z=\beta N /2^\z$; hence, only nodes in $\bar{D}_v^{2}$ can transmit. Consider an arbitrary node $w\in N(v)$. For each node $u\in \bar{D}_v^{2}$,
$d(u,w)\ge d(u,v)-d(v,w)\geq d(u,v)/2$.
Then, the interference at $w$ is at most:
\begin{equation*}
\cI_w \le \sum_{u\in \bar{D}_v^{2}}\frac{P}{d(u,w)^\z}\le 2^\z\sum_{u\in \bar D_v^{2}}\frac{P}{d(u,v)^\zeta}\leq 2^{\z} \Ic\le ((1-\e)^{-z}-1)N,
\end{equation*}
which implies that node $w$ receives $v$'s transmission: $\frac{P/d(v,w)^\z}{N+\cI_w}\ge\frac{\beta N(1-\e)^{-\z}}{N+((1-\e)^{-\z} - 1)N}\ge\beta$.
\end{proof}

\textit{The UDG and UBG Models.} These models are described by geometric  graphs: a node $u$ receives a message from another node $v$ if and only if $v$ is the only transmitting neighbor of $u$. 

In the Unit Disk Graph (UDG) and Unit Ball Graph (UBG) models the nodes are located in a metric space and two nodes are connected by an edge if and only if their distance is at most $\R$. The functionality of {\scc} can be modeled as follows: the transmission of a node $v$ is received by all its neighbors if there is no other node at distance less than $2\R$ from $v$ transmitting simultaneously, i.e. we can set the parameters to $\Ic=\infty$ and $\rho_c=2$.

\textit{The Quasi-UDG Model.}
The Quasi-UDG model is an extension of the UDG model: a) if $d(u,v) \le \R$ then $u$ and $v$ are connected by an edge, b) if $d(u,v) > \R'$ then they are disconnected, c) otherwise, $u$ and $v$ may be connected or not. In this case {\scc} may be implemented by setting $\Ic=\infty$ and $\rho_c=(\R+\R')/\R$, with the adversary constrained to follow the specific static situation captured by the QUDG.

\textit{The Protocol Model.} In the  Protocol Model, the nodes are in a metric space and there are two radii: $\R$ -- the communication radius, and $\R'$ -- the interference radius. A node $v$ receives the transmission of a node $u$ if and only if: 1) $v$ is in the communication range of node $u$, i.e. $d(u,v) \le \R$, and 2) there is no transmitting node $w$ such that $v$ is in the interference range of $w$: for each transmitting node $w\neq u$, $d(w,v) > \R'$.
{\scc} may be implemented here by setting $\Ic=\infty$ and $\rho_c=(\R+\R')/\R$.

\textit{The BIG Model.}
In the Bounded Independence Graph (BIG) model, for a parameter $\lambda$, 
we are given a graph on the nodes with the property that for every node $v$ and every $k \ge 1$,
the maximum independent set in the $k$-neighborhood of $v$ is $O(k^{\lambda})$.
The shortest-path distance metric on the graph is now naturally a $(1,\lambda)$-bounded independence metric.
To fit in our model, the growth parameter $\lambda$ must be less than $\zeta$.

\textit{$k$-hop Variants.}
These graph models can be naturally generalized to a model on interference, where 
nodes of distance at most $k$ cause interference, for some $k > 1$.
We capture this by extending $\rho_c$ as needed.

\paragraph{Implementing primitives with physical carrier sensing}

One way of implementing the primitives mentioned in this paper is to use physical carrier sensing, i.e. we assume the nodes have technology to detect if the interference (plus noise) is higher than a given threshold. 

We show below how to implement  {\ack}, {\cd} and {\ntd} primitives using carrier sensing.

\textbf{{\cd} primitive.}
The {\cd} primitive can be implemented using a carrier sensing threshold $T=P/((1-\e)\R)^\z$ and setting the parameter $\Icd<T$; {\busy} channel is detected if and only if the interference is at least $T$.

We will need the following technical fact.
\begin{lemma}\label{ieq}
For every $x_i \in [0, \frac{1}{2}]$, $i=1,2,\dots,n$, it holds that
$
4^{-\sum_{i=1}^{n} x_i} \leq \prod_{i=1}^{n} (1-x_i) \leq e^{-\sum_{i=1}^{n} x_i}.
$
\end{lemma}

\begin{proposition}
If $P_t(v)>\phi \ge 1$ in round $t$ then all nodes in $B(v, \R/2)$ detect {\busy} channel with probability at least $1-(1+2\phi) e^{-\phi}$. In particular, we can take $h_1=2$ if $\phi \ge 10$.
\end{proposition}
\begin{proof}
By the setting of $T$ and the definition of $B(v,\R/2)$, if two nodes in $B(v,\R/2)$ transmit in round $t$, then all nodes in $B(v,\R/2)$ will detect {\busy} channel.
Hence, the probability of all nodes in $B(v,\R/2)$ detecting {\busy} channel is at least the probability of more than one node transmitting in round $t$. The probability of no node transmitting is $p_0=\prod_{u\in B(v,\R/2)}(1-p_t(u)) \le e^{-\phi}$ by Lemma~\ref{ieq}. The probability of exactly one node transmitting is:
\[
p_1 = \sum_{u\in B(v,\R/2)}{p_t(u)\prod_{w\in B(v,\R/2)\setminus u}(1-p_t(w))}\le 2\sum_{u\in B(v,\R/2)}{p_t(u)\prod_{w\in B(v,\R/2)}(1-p_t(w))}\le 2\phi e^{-\phi},
\]
where we used the assumption that $p_t(u)\le 1/2$ and Lemma~\ref{ieq}.
Thus, the probability of detecting {\busy} channel is at least $1-p_0-p_1\ge 1-(1+2\phi)e^{-\phi}$.
\end{proof}

\begin{proposition}
For every $\rho>0$, if $P_t^\rho(v)<\eta$ and $\cI_t^\rho(v) < \Icd$ then $v$ detects {\idle} channel with probability at least $4^{-\eta}$.
\end{proposition}
\begin{proof}
By the setting of the threshold $T$, if there is no node transmitting in $D_v^\rho$ then node $v$ will detect {\idle} channel. Thus, the probability that $v$ detects {\idle} channel is at least
\begin{equation*}
\prod_{u\in D_v^{\rho}}(1-p_t(u))
\geq 4^{-\sum_{u\in D_v^{\rho}}p_t(u)}
\geq 4^{-\eta},
\end{equation*}
using Lemma~\ref{ieq}.
\end{proof}

\textbf{{\ack} primitive.}
In order to detect successful transmission, we can use interference threshold $T=\min\{\Ic, P/(\rho_c\R)^\z\}$ and set $\Iack<T$, where $\Ic(\e)$ and $\rho_c(\e)$ are the parameters of {\scc}. If a node $v$ senses that the interference is no higher than $T$, it knows that: 1. there is no node in $D_v^{\rho_c}$ transmitting, as otherwise the interference would be at least $P/(\rho_c\R)^\z$, 2. the interference is at most $\Ic$; thus, it knows that its transmission has been received by all neighbors in $N_t(\e)$ by {\scc}.

\textbf{{\ntd} primitive.}
If a node $v$ receives a message from a node $u$ then $v$ can separate the signal from the interference and  measure the received signal strength. As the nodes use uniform power assignment, node $v$ knows that $u\in D_v^{\e/2}$ if the received signal is stronger than $P/(\e\R/2)^\z$.

\paragraph{Implementing primitives by other means} The primitives can frequently be implemented in other ways, often with the logarithmic blowup that explains the differences
with the best carrier-sense-free results.

\textbf{{\cd} primitive.} In an asynchronous system, it may be impossible to implement {\cd} by other means than carrier sense. In a synchronized system, however, we can be achieved with logarithmic or polylogarithmic factor overhead.
Consider a given round. For each probability $p=2^{-i}, i=1,2, \ldots, \log n$, 
repeat $C \log n$ times: the senders in the original round transmit with probability $p$.
Using concentration bound with $C$ sufficiently large, one can infer the contention within an small approximation,
with high probability. Such a strategy has been applied, e.g., in \cite{HM12P}.

\textbf{{\ack} primitive.} 
A simple strategy is to work with only probabilistic guarantees of a transmission being received by all neighbors.
Then, simply repeat the protocol until this has been achieved $C \log n$ times, which gives an {\ack} guarantee, w.h.p.
This approach underlies, e.g., the local broadcast algorithms without carrier sense \cite{GMW08,HM11,YHWL12}.

\textbf{{\ntd} primitive.} This primitive, which is essential for dominator-based strategies for broadcast,
can be implemented using power control: by lowering the power on all units appropriately, one can ensure that nodes further away (by a small constant factor) will not be able to hear the message due to the ambient noise term, see e.g.~\cite{YHWTL12}. Alternatively, one can assume that distances can be determined in other ways, such as by GPS \cite{JKRS13,JKS13}.

\section{Proof of Proposition \ref{pr:Pgoodroundl}: Contention Control}\label{sec:contentiondetails}

Recall that we need to prove the following. We assume that at the beginning of the \emph{first} phase under consideration, the contention in the whole network is bounded by a constant. This holds for all algorithms in this paper, as the initial probability of nodes is always at most $1/n$.

\vspace{5pt}
\noindent\textbf{Proposition \ref{pr:Pgoodroundl}.}{
Let $\sigma\in (0,1)$.
If constants $\rho=\rho(\sigma,\hcI),\hphi = \hphi(\rho,\sigma)$ and $\gamma=\gamma(\rho,\hphi,\sigma,\hcI)$ are large enough then  for each node $v$ and phase $H$, with probability $1-O(n^{-3})$, 
a $(1-\sigma)$-fraction of the rounds in $H$ are good.
}
\vspace{5pt}

Recall that in a good round, there should be both bounded contention and low interference. The proof is split into two parts, each handling one of these properties.
Prop.~\ref{pr:Pgoodroundl} follows by simply combining those two parts.

 We will need the following concentration bounds.

\begin{lemma}\cite{SRS08,IK10}\label{le:chernoff}
Consider a collection of binary random variables $X_1,\ldots,X_n$, and let $X=\sum_{i=1}^nX_i$. If there are probabilities $p_1,\ldots,p_n$ with $E[\prod_{i\in S}X_i]\leq \prod_{i\in S}p_i$ for every set $S\subseteq\{1,\ldots,n\}$, then it holds for $\mu=\sum_{i=1}^np_i$ and $\delta>0$ that
\begin{equation*}
Pr[X\geq (1+\delta)\mu]\leq\left(\frac{e^{\delta}}{(1+\delta)^{1+\delta}}\right)^\mu\leq e^{-\frac{\delta^2\mu}{2(1+\delta/3)}}.
\end{equation*}
If, on the other side, there are probabilities $p_1,\ldots,p_n$ with $E[\prod_{i\in S}X_i]\geq \prod_{i\in S}p_i$ for every set $S\subseteq\{1,\ldots,n\}$, then it holds for $\mu=\sum_{i=1}^np_i$ and $0<\delta<1$ that
\begin{equation*}
Pr[X\leq (1-\delta)\mu]\leq\left(\frac{e^{-\delta}}{(1-\delta)^{1-\delta}}\right)^\mu\leq e^{-\delta^2\mu/2}.
\end{equation*}
\end{lemma}

\subsection{Bounded Contention Rounds}\label{sec:lowc}

First, we show that for each fixed node $u$, the contention in the \emph{local} neighborhood $B(u,\R/2)$ is bounded in most of the rounds of a phase. To this end, we show that in each round, the contention is either already low or will be halved with significant probability, then apply a concentration bound to show the claim.  Recall the constant $h_1$ from the definition of {\cd}.

\begin{lemma}\label{le:phi}
Let $H$ be a time interval and assume  that $P_t(u) =\phi_0\ge 0$ at the beginning of $H$. Then for every $\phi\ge 3$, $P_t(u)>\phi$ happens at most $O(\phi^{-k})\cdot |H| + \max\{0, \log_{6/5}(\phi_0/\phi)\}$ times during $H$, with probability $1-2^{-\mu |H|}$, where $\mu=h_1^{-\phi}/(1-O(\phi^{-k}))$ and $k>0$ is the edge change parameter.
\end{lemma}
\begin{proof}

By the assumptions on edge changes, we have that during phase $H$, the fraction of rounds where more than $\phi/4$ nodes become a neighbor of $u$ because of edge changes is  $O(\phi^{-k})$. Let $H'$ denote the remaining set of rounds. We have $|H'| = (1-O(\phi^{-k}))|H|$. In each of those rounds, the contribution of edge changes in $P_t(u)$ is clearly at most $\phi/8$. 

Next we bound the number of rounds in $H'$, where $P_t(u) > \phi$.
Consider such a round $t\in H'$. Let $\mu'=h_1^{-\phi}$. By the definition of {\cd}, all nodes in $B(u,\R/2)$ detect {\busy} channel in round $t$ and halve their transmission probabilities with probability at least $1-\mu'$. Thus, $Pr[P_{t+1}(u)\leq \frac{1}{2}P_t(u) + 1/2 + \phi/8]\ge 1-\mu'$, where the additive $1/2$ accounts for the sum of probabilities of the nodes that just join the ball $B(u,\R/2)$ due to node churn (recall that each of them has an initial probability at most $1/2n$) and $\phi/8$ is an upper bound on the contention due to edge changes (by the definition of $H'$).
For each round $t\in H'$, let us define a binary random variable $X_t$ as follows: $X_t=0$ if $P_{t}(u)\leq\phi$ or $P_{t}(u)>\phi$ and $P_{t+1}(u)\leq\frac{1}{2}P_t(u) +1/2 + \phi/8\le 5P_t(u)/6$ (recall that $\phi \ge 3$) and $X_t=1$ otherwise. By the discussion above, we have:
\[
\begin{aligned}
Pr[X_t=0]&=Pr[P_{t}(u)\leq\phi]+Pr[P_{t}(u)>\phi\wedge P_{t+1}(u)\leq\frac{5}{6}P_t(u)]\\
&=Pr[P_{t}(u)\leq\phi]+Pr[P_{t}(u)>\phi]\cdot Pr[P_{t+1}(u)\leq\frac{5}{6}P_t(u)|P_{t}(u)>\phi]\\
&\geq Pr[P_{t}(u)\leq\phi]+Pr[P_{t}(u)>\phi]\cdot (1-\mu')\\
&\geq 1-\mu'.
\end{aligned}
\]
 This implies that for each round $t\in H$ and each subset $S\subseteq H$ of earlier rounds ($s<t$ for $s\in S$), $Pr[X_t=0|\prod_{s\in S}X_s=1]\geq 1-\mu'$ and $Pr[X_t=1|\prod_{s\in S}X_s=1]\leq \mu'$. Thus, for each subset $S\subseteq H'$ we get 
$
E[\prod_{s\in S}X_s]\le \mu'^{|S|},
$
 and if we denote $X=\sum_{t\in H'}{X_t}$, then $E[X]\le \mu' |H'|=\mu |H|$, where we denote $\mu = \mu' |H|/|H'|$. Thus, we can apply Chernoff bound (Lemma~\ref{le:chernoff}) with $\delta=3/2$ to bound $X$ with high probability:
 \begin{equation}\label{eq:Xupperl}
Pr[X \ge 3\mu|H|/2]\le 2^{-\mu|H|}.
\end{equation}
Using this bound, we obtain a bound on the number of rounds $t$ with $P_t(u)> \phi$. Let $Y$ denote this number.
Consider a maximal interval $\hat H=[r_1,r_2]\subseteq H'$ such that $P_t(u) > \phi$ for all $t\in \hat H$. For each round  $t\in \hat H$, $P_{t+1}(u)\le \frac{5}{6}P_t(u)$  if $X_t=0$ and $P_{t+1}(u)\leq 2P_t(u) + 1 + \phi/4 < 3P_t(u)$ otherwise. By maximality of $\hat H$, if $r_1$ \emph{is not} the first round of phase $H$, then we have $P_{r_1 - 1}(u) \le \phi$. Then a simple calculation shows that $X_t=1$ for at least $1/8$ part of the rounds of $\hat H$, i.e. $\hat X=\sum_{t\in \hat H}X_t \ge |\hat H|/8$. On the other hand, if there is a unique maximal interval $\hat H$ starting at the first round of phase $H$ and such that $P_{r_1}(u) > \phi$, then it must hold that $\hat X \ge (|\hat H| - L)/8$, where $L=\max\{0, \log_{6/5}(\phi_0/\phi)$. Combining these observations, we get that $Y\le 8X+L$.

By combining this bound with~(\ref{eq:Xupperl}), we have:
\[
Pr[Y\ge 8\cdot 3\mu|H|/2+L] \le Pr[X\ge \delta\mu|H|] < 2^{-\mu|H|}.
\]
Thus, with probability $1-2^{-\mu|H|}$, in at most $O(h_1^{-\phi} + \phi^{-k})=O(\phi^{-k})$ fraction of rounds in $H$ there can be $P_t(u) > \phi$.
\end{proof}

Recall that for node $v$ we need to show that $P_t^\rho(v)\le \hphi$ for most rounds in a phase. We prove this by taking a \emph{constant size} $\R/2$-cover of $D_v^\rho$ and applying Lemma~\ref{le:phi} to all nodes in the cover \emph{simultaneously}.

\begin{proposition}\label{pr:lowc}
For every $\sigma\in (0,1)$ and $\rho>0$, there are constants  $\hphi = \hphi(\rho,\sigma)$ and $\gamma =\gamma(\rho,\sigma)$, such that for every node $v$ and phase $H$, with probability $1-O(n^{-3})$, at least $(1-\sigma)$-fraction of rounds in $H$ are bounded contention rounds, i.e. $P_t^\rho(v) \le \hphi$.
\end{proposition}
\begin{proof}
It will be sufficient to choose $\gamma$ and $\hphi$ such that $\hphi = O(\rho^{\l + \l/k}/\sigma^{1/k})$ and $\gamma=h_1^{O(\hphi/\rho^\l)}$.
 By the definition of $\lambda$ and bounded independence, there is a $\R/2$-cover $S$ of $D_v^\rho$ of size $|S|=O(\rho^\l)$. Set $\phi = {\hphi} / |S|$. Let us fix a node $u\in S$. If $H$ is the first phase or $P_t(u)\le \phi$ then we set $\hat H=H$. Otherwise, consider the phase $H'$ of $\gamma\log n$ rounds preceding $H$. By Lemma~\ref{le:phi}, if we set $\gamma\ge 3h_1^\phi$ then, with probability $1-n^{-3}$, there is a round in $H'$ where $P_t(u) \le \phi$. Let $t_0$ be such a round and set $\hat H = [t_0,t_l]$, where $t_l$ is the last round of $H$. Clearly, $|\hat H| \le 2 |H|$. Now we can apply Lemma~\ref{le:phi} for node $u$ with $\phi$ as above and conclude that with probability at least $1-n^{-3}$, there are at most $O(\phi^{-k})\cdot|\hat H|$ rounds in $\hat H$ (note that $\phi_0=\phi$) where $P_t(u) > \phi$, so there are at most $O(\phi^{-k})\cdot|H|$ rounds in $H$ where $P_t(u) > \phi$. Now we can apply the same argument for all nodes in $|S|$ simultaneously and conclude that with probability at least $1-O(n^{-3})$, there are at most $O(\phi^{-k})\cdot|S||H|\le \sigma|H|$ rounds in $H$ where $P_t(u) > \phi$ for \emph{every} $u\in S$, if we take $\hphi = O(\rho^{\l + \l/k}/\sigma^{1/k})$ to be large enough. In the remaining $(1-\sigma)|H|$ rounds we have $P_t^\rho(v)\le \sum_{u\in S}P_t(u)\le \hphi$.
\end{proof}

We extract the following result from the proof of Lemma~\ref{le:phi}, to use it later. The proof is similar to the proof of Prop.~\ref{pr:lowc}.

\begin{corollary}\label{co:Pexpectation}
For every $\phi\ge 3$ and each node $v$, the expected number of rounds in each phase $H$ where $P_t(v) > \phi$, is  $O(\phi^{-k})\cdot|H|$.
\end{corollary}
\begin{proof}
Let $\hat H$ be the shortest time interval containing $H$ such that $P_t(v)\le \phi$ at the beginning of $\hat H$.  Let $\mathcal{E}$ denote the event that $|\hat H| \le 2 |H|$. As in the proof of Prop.~\ref{pr:lowc}, we have $Pr[\mathcal{E}] > 1-n^{-3}$.  Let $X$ denote the number of rounds $t\in \hat H$ with $P_t(v)>\phi$. We know from the proof of Lemma~\ref{le:phi} that $E[X|\mathcal{E}] = O(\phi^{-k})\cdot|\hat H| = O(\phi^{-k})\cdot |H|$ and $E[X]=E[X|\mathcal{E}]\cdot Pr[\mathcal{E}] + E[X|\bar{\mathcal{E}}]\cdot Pr[\bar{\mathcal{E}}]= O(\phi^{-k})\cdot |H|$. This completes the proof as $H\subseteq \hat H$.
\end{proof}

\subsection{Low Interference Rounds}

Next we show that if $\rho$ is appropriately chosen then for each node $v$, $\bcI_t^\rho(v) < \hcI$ happens most of the time during each phase w.h.p. In order to show this, first we split the set of nodes in $\bar D_v^\rho$ into local neighborhoods. We show that if in a given round the contention in each local neighborhood is bounded by an appropriate threshold, then the expected interference is small. For each local neighborhood, the expected number of rounds when the contention is higher than its threshold is bounded using Corollary~\ref{co:Pexpectation}. Combining these results into one we get a bound on the expected number of rounds when $\bcI_t^\rho(v)$ is small. It then remains to apply a concentration bound to conclude the proof.


\begin{proposition}\label{pr:Pinterferenceboundl}
For every $\sigma\in (0,1)$ and $\hcI > 0$, there are constants  $\rho = \rho(\sigma, \hcI)$ and $\gamma =\gamma(\sigma)$, such that for each node $v$ and  phase $H$,  with probability $1-O(n^{-3})$, at least $(1-\sigma)$-fraction of rounds in $H$ are low interference rounds, i.e.  $\bcI_t^\rho(v)\leq \hcI$.
\end{proposition}
\begin{proof}
Let $D_i$ denote the in-ball of radius $\rho\R+i\R$ centered at node $v$ for $i=0,1,\dots$. Let $S_i\subseteq D_i$ be an $\R/2$-cover of $D_i$ of size $|S_i|\leq C\cdot (\rho+i)^{\l}$ for a constant $C=C(\l)$, which exists by bounded independence. Let $\phi_i=(\rho+i)^{(\z-\l)/2}\phi$ for $i=1,2,\dots$, where $\phi = \Omega(\l/\sigma)$ is a large enough constant to be defined later.

\begin{claim}\label{le:PexpectAl}
If $P_t(u)<\phi_i$ holds for all $i$ and $u\in S_i\setminus S_{i-1}$ and $\rho \ge \left(\frac{2C \phi P (\z+\l)}{R^\zeta\hcI(\z - \l)}\right)^{2/(\z-\l)}$ then $\bcI_t^\rho(v) \le \hcI$.
\end{claim}
\begin{proof}

We have, by the definition of expected interference (explanations below),
\begin{equation*}
\begin{aligned}
\bcI_t^\rho(v)=\sum_{w\in \bar D_v^{\rho}}\frac{P}{d(w,v)^\z}\cdot p_t(w)& \leq \sum_{i\geq 1}\frac{P}{((\rho+i-1) \R)^\z}\cdot |S_i\setminus S_{i-1}|\cdot \phi_i \\
&\leq \frac{\phi P(1+1/\rho)}{R^\zeta}\sum_{i\geq 1}\frac{|S_i\setminus S_{i-1}|}{(\rho+i)^{(\z+\l)/2}}\\
&\leq \frac{2\phi P}{R^\zeta}\sum_{i\geq 1}|S_i|\cdot (\frac{1}{(\rho+i)^{(\z+\l)/2}}-\frac{1}{(\rho+i+1)^{(\z+\l)/2}})\\
&\leq \frac{2\phi P}{R^\zeta}\sum_{i\geq1}C(\rho+i)^\l\cdot\frac{(\z+\l)/2}{(\rho+i)^{(\z+\l)/2+1}}\\
&= \frac{C \phi P (\z+\l)}{R^\zeta} \sum_{i\geq1}(\rho+i)^{-(1+(\zeta-\lambda)/2)}\\
&\leq C \phi (\z+\l)\cdot (P/R^{\zeta})\cdot 2/(\zeta-\lambda)\cdot \rho^{-(\zeta-\lambda)/2},
\end{aligned}
\end{equation*}
where the first inequality follows from the definition of sets $S_i$ and the assumption of the claim, the second one follows by 
$\phi_i=(\rho+i)^{(\z-\l)/2}\phi$ and $(\rho+i-1)(1+1/\rho) \ge \rho + i$, the third one is a rearrangement of the sum, the fourth one follows by the fact that $\frac{1}{(t-1)^{\alpha}}-\frac{1}{t^\alpha}< \frac{\alpha}{(t-1)^{\alpha+1}}$ for every $t> 1$ and $\alpha\ge 1$ (here, $(\z+\l)/2 > 1$), and the last one follows because $\z-\l >0$. Thus, if $\rho \ge \left(\frac{2C \phi P (\z+\l)}{R^\zeta\hcI(\z - \l)}\right)^{2/(\z-\l)}$ then the claim follows.
\end{proof}

Let us put the nodes in an arbitrary order and let $Z_{i,j}$ denote the number of rounds $t \in H$ s.t. $P_{t}(u_j)\geq \phi_i$, where $u_j$ is the $j$-th node in $S_i\setminus S_j$. 
By Corollary~\ref{co:Pexpectation},
$
E[Z_{i,j}]= O(\phi_i^{-k})\cdot|H|
$
holds for each $i,j$. 
Let $Z$ denote the number of rounds in $H$ when $\bcI_t^\rho(v)> \hcI$. By Claim~\ref{le:PexpectAl}, $E[Z]$ can be bounded as follows,
\[
E[Z]\leq\sum_{i}\sum_{j}E[Z_{i,j}]
\leq \sum_i|S_i\setminus S_{i-1}|\cdot O(\phi_i^{-k})\cdot |H|
\leq \frac{\sigma}{3}|H|,
\]
where the last inequality holds if $\phi=C''/\sigma^{1/k}$ for a large enough constant $C''$; we have $|S_i\setminus S_{i-1}|\le C\l(\rho+i)^{\l-1}$ and $\phi_i^{-k}\le \phi^{-k}\cdot (\rho+i)^{k(\z-\l)/2}$ with $k(\z-\l)/2 > \l$ (by definition), so we have $\sum_i|S_i\setminus S_{i-1}|\cdot O(\phi_i^{-k})=O(\phi^{-k})$. 

Each $Z_{i,j}$ can be considered as the sum of binary random variables similar to that in the proof of Lemma~\ref{le:phi}. Thus, $Z$ can also be seen as the sum of binary random variables. Moreover, the binary random variables satisfy the conditions of the Chernoff bound in Lemma~\ref{le:chernoff}. Thus, by taking $\delta=2$ and  $\gamma \ge 9/\sigma$, we have:
$
Pr[Z\geq 3\cdot\frac{\sigma}{3}|H|]\leq 2^{-\sigma|H|/3}\le n^{-3}.
$
\end{proof}

\section{Proof of Proposition~\ref{pr:Pdecrease}: Transmissions in Good Rounds}

\begin{lemma}\label{le:broadgood}
Let $\rho\ge\rho_c+2$, $\hcI \le (1-1/\rho)^\z \Ic/10$ and $\phi >0$, where $\rho_c,\Ic$ are the parameters of {\scc}. Let $t$ be a round such that for a node $u$, $P_t^\rho(u) < \phi$ and $\hcI_t^\rho(u) < \hcI$. \\
If a node in $B(u,\R/2)$ transmits in round $t$ then it mass-delivers, with probability at least $0.9\cdot 4^{-\phi}$.
\end{lemma}
\begin{proof}
Let $v$ be the transmitting node, which could possibly be $u$ itself. The probability that no other node in $D_u^{\rho}$ transmits is 
\[
p'=\prod_{w\in D_u^{\rho}\setminus\{v\}}(1-p_t(w))\ge 4^{-\sum_{w\in D_u^{\rho}\setminus\{v\}}p_t(w)}\ge 4^{-\phi},
\]
where we used Lemma~\ref{ieq} with the assumption that $p_t(w)\le 1/2$ and the bounded contention assumption.

On the other hand, we have that $Pr[\cI_t^{\rho-1}(v) \le \Ic]\ge 9/10$. Indeed, $Pr[\cI_t^\rho(u) \le \cI'=(1-1/\rho)^\z \Ic]\ge 9/10$ follows by the bounded interference assumption and Markov inequality, so it suffices to show that $\cI_t^{\rho-1}(v) \le (1-1/\rho)^{-\z} \cI_t^\rho(u)$, which we have because for each node $x\in\bar D_u^{\rho}$,  $d(x,v)\geq d(x,v)-d(v,u)>(1-1/\rho)d_{xu}$, implying that 
\[
\cI_t^{\rho-1}(v)=\sum_{x\in\bar D_u^{\rho}}P/d(x,v)^\z < \sum_{x\in \bar D_u^{\rho}}(1-1/\rho)^{-\z}\cdot P/d_{xu}^\z=(1-1/\rho)^{-\z}\cI_t^\rho(u).
\]
Thus, we have that with probability $0.9\cdot 4^{-\phi}$, $v$ transmits, no other node in $D_u^\rho$ transmits and the interference from $\bar D_u^\rho$ at $v$ is at most $\Ic$, so $v$ delivers its message, by {\scc}.
\end{proof}

\noindent\textbf{Proposition \ref{pr:Pdecrease}}
\textit{
Assume that constants $\hphi, \rho,\gamma$ are large enough. 
For each node $v$ and phase $H$, 
if at least $1/10$-fraction of the rounds of $H$ are of high contention,
then $\Omega(|H|)$ nodes in $D_v^{\rho}$ mass-deliver, with probability $1-O(n^{-3})$.
}
\vspace{5pt}
\begin{proof}
Let $\sigma=1/10$. Let us choose constants $\hphi, \rho',\gamma$ such that Prop.~\ref{pr:Pgoodroundl} holds for $\sigma'=\sigma/2$ and $\hcI_1 \le  (1/2-1/\rho')^\z \Ic/20$ and such that $\rho' > 2(\rho_c+2)$ is large enough, as stated below, where $\rho_c,\Ic$ are the parameters of {\scc}. We will use the fact that in Prop.~\ref{pr:Pgoodroundl}, $\hphi$ has the following dependency on $\rho'$: $\hphi = O(\rho'^{\l + \l/k})$, where $k$ is the edge change parameter (see the proof of Prop.~\ref{pr:lowc}).
We set $\rho=\rho'/2$.
Namely, if we choose $\hphi, \rho,\gamma$ large enough (satisfying the relation above), then we have that with probability $1-O(n^{-3})$, there are at least $(1-\sigma/2)|H|$ rounds $t\in H$ such that $P_t^{2\rho}(v) < \hphi$ and $\hcI_t^{2\rho}(v) < \hcI_1$. Let $t$ be such a round.
\begin{claim}
For each node $u\in D_v^\rho$, it holds that $P_t^\rho(u)< \hphi$ and $\hcI_t^\rho(u)< \hcI_2=(1-1/\rho)^\z\Ic$.
\end{claim}
\begin{proof}
Let us fix a node $u\in D_v^\rho$. First note that $D_u^\rho \subseteq D_v^{2\rho}$, implying that $P_t^\rho(u)\le P_t^{\rho'}(u)< \hphi$. It remains to bound the expected interference at $u$ by nodes in $\bar D_u^\rho$. This interference can be split into two parts: 1. the interference by nodes in $\bar D_v^{\rho'}$ and 2. the interference by nodes in $D_v^{\rho'}\setminus D_u^\rho$. The first part can be bounded by $2^\z\hcI_1=\hcI_2/2$ using a computation as in the proof of Lemma~\ref{le:broadgood}. The second part can be bounded as follows:
the contention in the area $D_v^{\rho'}\setminus D_u^\rho$ is at most $\hphi$ and the distance of those nodes to node $u$ is at least $\rho$, so the expected interference by nodes in $D_v^{\rho'}\setminus D_u^\rho$ is at most 
$\hphi \cdot \frac{P}{(\rho\R)^\z}=O(\frac{\rho^{\l+\l/k}}{\rho^\z})=O(\frac{1}{\rho^{(\z-\l)/2}})$, and is less than $\hcI_2/2$ if $\rho$ is large enough, as $\z-\l>0$. Here we used the assumption that $k>2\l/(\z-\l)$. 
\end{proof}
Recall that there are at least $\sigma|H|$ high contention rounds for node $v$, i.e. rounds where $P_t^\rho(v) > \eta$. Let $t$ be such a round. By the bounded independence property, there is a $\rho\R$-cover $S$ of $D_v^\rho$ of constant size. By the pigeonhole principle, there is a node $u\in S$ such that $P_t^\rho(u) > \eta' = \eta/|S|$. 

By summarizing the above said, we get that there are at least $(1-\sigma/2)|H|$ rounds $t$ s.t. for all nodes $u\in D_v^\rho$, $P_t^\rho(u) < \hphi$ and $\hcI_t^\rho(u) < \hcI_2$. On the other hand, there are at least $\sigma|H|$ rounds where there is a node in $D_v^\rho$ that has $P_t^\rho(u) > \eta'$. Thus, there are at least $\sigma|H|/2$ rounds $t$ with the following property: there is a node $u\in D_v^\rho$ such that $P_t^\rho(u) \in (\eta', \hphi)$ and $\hcI_t^\rho(u) < \hcI_2$. Fix such round $t$ and node $u$. By Lemma~\ref{le:broadgood}, there is a node $u'\in B(u,\R/2)$ that transmits (the probability of this is at least $\eta'$) and delivers its message with probability $0.9\eta' 4^{-\hphi}$. Thus, the expected number of nodes in $D_v^\rho$ that deliver their message during $H$ is $\Omega(4^{-\hphi}|H|)$. The proof now follows by applying a Chernoff bound with the assumption that the constant $\gamma$ is large enough: $\gamma=\Omega(4^{\hphi})$.
\end{proof}

\section{Proof of Proposition \ref{pr:idle}: Low Contention Rounds}

\noindent\textbf{Proposition~\ref{pr:idle}}
\textit{
Assume that $\hphi, \rho,\gamma$ are large enough. 
For each node $v$ and phase $H$, if at least $9/10$-fraction of the rounds of $H$ are low contention rounds, 
then  with probability $1-O(n^{-3})$, in at least $3/5$-fraction of the rounds of $H$, $v$ will detect {\idle} channel and have low contention and low interference.
}
\vspace{5pt}
\begin{proof}
Let $\sigma=9/10$.
Assume that $\hphi, \rho,\gamma$ large enough, such that Prop.~\ref{pr:Pgoodroundl} holds for $\sigma'=\sigma/10$ and $\hcI \le  \Icd/10$, where $\Icd$ is the parameter of {\cd}.
Let $\mathcal{E}$ denote the event that there are at least $(1-\sigma/10)|H|$ good rounds for $v$ in phase $H$. By Prop.~\ref{pr:Pgoodroundl}, we have $Pr[\mathcal{E}] = 1-O(n^{-3})\ge 19/20$, if $\gamma$ is large enough. Given $\mathcal{E}$, there are at least $9\sigma|H|/10$ rounds $t$ that are both good and low contention for $v$. Let $H'$ denote the set of such rounds and  let $X_t$ be a binary random variable with value 1 if and only if $v$ detects {\idle} channel in round $t\in H'$. Let $X=\sum_{t\in H'}X_t$. For each round $t\in H'$, since $t$ is a good round and $\hcI < \Icd/10$, it holds with probability at least $9/10$ that the interference at node $v$ from nodes in $D_v^\rho$ is at most $\Icd$. Given this and the fact that in low contention rounds we have $P_t^\rho(v) < \eta=\log_{h_2}(10/9)$, $v$ will detect {\idle} channel with probability at least $h_2^{-\eta}=9/10$ in round $t$. Thus, in each round $t\in H'$, irrespective of prior rounds, we have $Pr[X_t]\ge 81/100$.
  This implies: $E[X|\mathcal{E}]\ge (9^2/10^2)|H'|\ge
	\frac{9^2}{10^2}\cdot \frac{9^2}{10^2}|H|$ and 
	\[
	E[X] \ge (19/20)\cdot\frac{9^4}{10^4}|H| > 0.62|H|.
	\]
Just as in the proof of Lemma~\ref{le:phi}, we can now apply Chernoff bound to obtain that $Pr[X<3|H|/5] < n^{-3}$ if $\gamma$ is large enough. This completes the proof by recalling that all rounds of $X$ are low contention and low interference rounds for node $v$.
\end{proof}

\section{Proof of Thm.~\ref{thm:broadcast}: Non-Spontaneous Broadcast}

\noindent \textbf{Theorem \ref{thm:broadcast}}
\textit{
Assume the edge change rate $\tau$ is  sufficiently small. There are constants $\beta,c>0$, such that when running {\br}($\beta$) in the synchronous mode, each node $v$ receives the message in $O(D_{st}^c(s,v))$ rounds w.h.p.
}
\vspace{5pt}
\begin{proof}
Let us fix  constants $\rho,\hphi,\gamma$ so that Props.~\ref{pr:Pdecrease} and~\ref{pr:idle} hold with interference threshold $\hcI=\min\{(1-1/\rho)^\z\Ic(\e/2),\Iack(\e/2)\}/10$, where $\Iack(\e/2)$ and $\Ic(\e/2)$ are the parameters corresponding to  precision $\e/2$.  We set the passiveness parameter of {\ta} to $\beta=\gamma+5$ and the stable distance parameter $c=12\gamma$.
Consider a node $u$ that receives the message in round $t$. Let $v$ be such that $v\in N(u,\e)$  in the time interval $T=[t, t+c\log n]$ (and both stay alive during $T$).
The theorem follows by an induction and union bound from the claim below.
\end{proof}

\begin{claim}
Node $v$ gets the message of $u$ during $T$, with probability $1-O(n^{-2})$.
\end{claim}
\begin{proof}
We prove the Claim by contradiction, and assume that $v$ cannot get the message during $T$. Let us split $T$ into phases (for $u$) and, similar to the proof of Thm.~\ref{thm:localbroadcast}, classify the phases $H$ into types: (type A) at least $1/10$-th of rounds in $H$ are high contention rounds for $u$ -- $P_t^{\rho}(u)\ge \eta$, and (type B) at least $9/10$-th of rounds in $H$ are low contention rounds for $u$ -- $P_t^{\rho}(u) < \eta$.
We show that with probability $1-O(n^{-3})$, all phases in $T$ are of type B. Consider a type A phase $H$. We know from Prop.~\ref{pr:Pdecrease} that, with probability $1-O(n^{-3})$, there is a set $S$ of $\Omega(|H|)$ nodes in $D_v^{\rho}$ that deliver their messages and restart {\ta} during $H$. 
We split $S$ into three subsets, $S_1$ -- those inserted by churn during $H$, $S_2$ -- those inserted by edge change during $H$ and $S_3$ -- the rest. 

Consider $S_2$ first. We know that in each local neighborhood, there are at most $\tau |H|$ nodes arriving due to edge changes. Note that $D_u^\rho$ can be covered with $O(\rho^\l)$ local neighborhoods; hence, the total number of nodes arriving in $D_u^\rho$ due to edge changes is $O(\rho^\l)\cdot \tau|H|$. Since $|S|=\Omega(|H|)$, setting $\tau = O(\rho^{-\l})$ a small enough constant gives $|S_1| << |S|$.
Now consider $S_3$. By the {\ntd} primitive, for each node $u\in S_3$, all nodes in $D_u^{\e/2}$ also restart {\ta}, setting their transmission probability to $n^{-\beta}$. By the setting of $\beta= \gamma + 5$, the probability that a node restarting {\ta} in phase $H$ transmits again in $H$ is at most $n^{-5}$. Thus, with probability at least $1-O(n^{-3})$, the nodes in $S_2$ constitute an $\e\R/2$-packing of $D_v^{\rho}$ and are at most $|S_2|\le O(\rho/\e)^\l=O(1)$. As for $S_1$, the probability that a newly arriving node transmits during phase $H$ is at most $n^{-5}$, so $|S_1|=0$ w.h.p. 
These bounds lead to a contradiction -- $|S_1| + |S_2| + |S_3| < |S|$, if constant $\gamma$ is large enough. Thus, each phase $H$ of $T$ is of type B with probability $1-O(n^{-3})$. It remains to recall that $T$ contains constant number of phases.

Now assume that $H$ is of type B. In a similar way as in Theorem~\ref{thm:localbroadcast}, we can apply Prop.~\ref{pr:idle} and show that if $p_t(u)=p_0$ at the beginning of $H$, then by the end of $H$, the effect on $p_t(v)$ is equivalent to applying at least $|H|/10$ of $p_{t+1}(v)=\min\{2p_t(v),1/2\}$ operations on $p_0$. Since $p_t(u)$ at the beginning of $T$ is at least $n^{-\beta}$, it will take at most 11 phases to increase the probability to the value $1/2$. Then, a single phase will suffice for $u$ to deliver its message (including to node $v$), with probability at least $1-O(n^{-2})$.
\end{proof}

\section{Spontaneous Broadcast}
\label{sec:spontaneous}

The basic observation is that if all the nodes start running {\br}$^*$ simultaneously, a \emph{constant density dominating set} can be computed in $O(\log n)$ rounds (cf. \cite{SRS08}), where an $r$-dominating set of density $\kappa$ is a set $S$ of nodes such that for each node $u$, $1 \le |\{v\in S:u\in N(v,r)\}|\le \kappa$. The dominator algorithm is as follows: all nodes run {\br}$^*$ simultaneously, and 1. if a node $v$ stops by {\scc} then it is a \emph{dominator}, 2. if a node $u$ stops by detecting {\ntd} of node $v$ then it is dominated by $v$.
Having formed a constant density dominating set, it remains to disseminate the message using only dominators. 
 In order for the set of dominators to derive the connectivity properties of the original graph, we need that $(V,d)$ forms a $(\e\R/8,\l)$-bounded independence \emph{metric} space.
\begin{theorem}
In the static spontaneous setting, there is a uniform algorithm that performs broadcast in $O(D_G + \log n)$ rounds, w.h.p.
\label{thm:static-spont-bcast}
\end{theorem}

The algorithm for spontaneous broadcast consists of two stages: 1. compute a constant density $\e\R/4$-dominating set $DS$, 2. transmit the message using only the nodes in $DS$.
\paragraph{Dominating Set.} The dominating set is constructed by running {\br}$^*$ in the spontaneous mode, i.e. all nodes start running the algorithm simultaneously. Let $DS$ be the set of nodes that stop the algorithm by {\scc}. Since the dominated nodes use $\ntd(\e/2)$ as a stopping condition, $DS$ is a $\e\R/4$-dominating set. Moreover, $DS$ is a $\e\R/8$-packing, which implies that each in-ball $D$ of radius $k\cdot \R$ contains at most $O(k/\e)^\l$ dominators. In particular each node $u$ is dominated by at most constant number of dominators. The proof the algorithm terminates after at most $O(\log n)$ rounds is almost identical to the proof of Theorem~\ref{thm:broadcast}.
\paragraph{Broadcast.}
Recall that we assume that the communication graph is connected and has diameter $D_G$. Also, since $d$ is a metric, the graph is undirected.

The broadcast part is as follows: in the first round, the source node transmits the message to its neighbors; each dominator $u$, upon receiving the message, transmits it in each round with probability $p_t(u)=p_0$ until detecting {\ack}($\e/2$), where $p_0$ is a small enough constant.

Note that as soon as all dominators successfully transmit at least once, all nodes will get the message.
Moreover, if the constant $p_0$ is small enough, then the two algorithms can be run simultaneously: the key point is that the number of dominators in each $c\R$-neighborhood is bounded by a constant (not depending on $p_0$). This ensures that the nodes \emph{need not know $n$}, in order to coordinate the two algorithms.
It remains to show that the message will get to all dominators in $O(D_G+\log n)$ rounds using the broadcast algorithm.

Consider a  graph $H$ defined over the dominating set $DS$, where for every pair of nodes $u,v\in DS$, $u,v$ form an  edge if $v\in N(u,\e/2)$.
\begin{claim}
The diameter of $H$ is at most $D_G$.
\end{claim}
\begin{proof}
Note that for every pair of nodes $u,v\in G$  with $u\in N(v,\e)$, the corresponding dominators $u'$ and $v'$ are adjacent in $H$.
Indeed, by the definition of the dominating set, we have
\[
d(u',v')\le d(u',u) + d(u,v) + d(v,v')\le \e \R/4 + (1-\e)\R + \e \R/4 = (1-\e/2)\R.
\]
 Let $u',v'$ be two arbitrary nodes in $DS$ and let $P$ be the path of length at most $D_G$ in $G$ connecting nodes $u'$ and $v'$. By the observation above, if we replace each node $w\in P$ with its dominator $w'$, we obtain a path $P'$ of length $D$ in $H$, connecting $u'$ and $v'$. The completes the proof of the claim.
\end{proof}

\begin{claim}
For every transmitting node $w\in DS$, the probability that $w$ delivers its message to all its neighbors in $H$ is $\Omega(1)$ in each round $t$, if constant $p_0$ is sufficiently small.
\end{claim}
\begin{proof}
Let $S$ denote the set of neighbors of node $w$ in $H$, i.e. the nodes of $DS$ that are at distance at most $(1-\epsilon/2)\R$ from $w$. It suffices to show that all nodes in $S$ receive the message in a fixed round $t$ where $w$ transmits with $\Omega(1)$ probability.
This event holds if the interference at $w$ is no more than $\Ic$ -- event $\mathcal{E}_1$, and no other node transmits in $D_w^{\rho_c}$ -- event $\mathcal{E}_2$, in round $t$.

By the properties of the dominating set, the contention in a ball of radius $\rho_c\R$ is at most $O(p_0(\rho_c/\e)^{\l})$. Similarly to the proof of Prop.~\ref{pr:Pinterferenceboundl}, it can be shown that the expected interference by nodes in $\bar D_w^{\rho_c}$ is $O(p_0)$. Thus, by Markov inequality, $Pr(\mathcal{E}_1)=1-O(p_0)=\Omega(1)$ if $p_0$ is small enough. On the other hand, since the contention in $D_w^{\rho_c}$ is $O(p_0)$,  $Pr(\mathcal{E}_2)\ge 4^{-{O(p_0)}}$ follows by Lemma~\ref{ieq}. Thus, $Pr(\mathcal{E}) \ge Pr(\mathcal{E}_1)\cdot Pr(\mathcal{E}_2)=\Omega(1)$.
\end{proof}
Given that each informed node in $DS$ has  probability at least  $\eta=\Omega(1)$ of successfully broadcasting the message, the rest of the proof essentially follows along the lines of the proof of~\cite[Lemma 6]{BHM13P}. We present a sketch of the proof for completeness of the argument.

Let $S_t$ be the set of nodes in $DS$ that have been informed by round $t$, where $R_0$ contains only the source node. Let us fix a node $v\in DS$ and let $d_t$ be the distance (in graph $H$) from $v$ to the nearest node in $S_t$. Note that $d_0\le D_G$. The difference $\delta_t=d_{t-1}-d_t$ is the progress made in round $t$ and is a Bernoulli random variable with $E[\delta_t]\ge \eta$ for all $t$ when $d_t>0$ and $\delta_t=0$ otherwise. Let ${\tilde\delta}_t$ be a random variable that has the same distribution as $\delta_t$ when $d_t>0$ and is an i.i.d. Bernoulli random variable with mean $\eta$ otherwise. Let $\Delta_t=\sum_t{{\tilde\delta}_t}$. Note that node $v$ has been informed by round $t$ iff $\Delta_t\ge d_0$. Thus, we need to bound the probability $Pr(\Delta_t<d_0)$. Let $Z_t=\Delta_t-\eta t$. It is easy to show that the sequence $Z_t$ is a submartingale. Moreover, for a round $t\ge \frac{c}{\eta}(D_G+\log n)$ and constant $c>1$, $\Delta_t<d_0$ implies that $Z_t < -(c-1)(D_G+\log n)$. Since the sequence $Z_t$ is a submartingale, we can apply Azuma-Hoeffding bound to show that $Pr(\Delta_t<d_0) \le Pr(Z_t < -(c-1)(D_G+\log n)) < n^{-O(c)}$. Now the theorem follows by union bound over all nodes, by choosing the constant $c$ large enough.

\section{Proof of Theorem~\ref{thm:brlowerbound}: Necessity of {\ntd}}

\noindent\textbf{Theorem~\ref{thm:brlowerbound}}
\textit{For every (possibly randomized) broadcast algorithm $\mathcal{A}$ that uses neither node coordinates nor {\ntd} primitive, there is a $(\e\R/8, 1)$-bounded-independence metric space where $\mathcal{A}$ needs $\Omega(n)$ rounds to do broadcast in a $O(1)$-broadcastable network, even if the nodes have {\cd} and {\ack} primitives and operate spontaneously.}
\vspace{5pt}
\begin{proof}
Assume that $\z=2$. We present the construction using a distance function $d$: it is then straightforward to construct the corresponding path-loss matrix. Denote $\delta = \e/(8(1-\e))$.   First, assume the nodes operate non-spontaneously, i.e. a non-source node may need to receive a message in order to start participating in a protocol. Recall that $R_B=(1-\e)\R$. Consider $n$ points $p_1,p_2,\dots,p_n$, such that for every $i,j \le n-2$, it holds that $d(p_i,p_j)=\delta R_B=\e\R/8$, $d(p_i,p_{n-1})=\mu R_B < \R$, $d(p_i,p_n)=(\mu + 1)R_B > \R$ and $d(p_{n-1}, p_{n})=R_B$, where $\mu =\e (1+\e)/(1-\e) < 1$ (see the diagram in Fig. \ref{fig:nonspon}). Clearly, this set of points forms a $(\e\R/8,1)$-bounded independence metric space. We place $n$ wireless nodes  at distinct points $p_1,p_2,\dots,p_n$ uniformly at random. Let $v_i$ be the node at $p_i$. We assume, further, that communication only happens according to {\scc}; in particular, if the interference at a node is more than $\Ic$ then \emph{none} of its neighbors receives its transmission.

Note that broadcast in this network can be completed in 2 steps, starting at any point.
Note also that $v_n$ cannot be directly reached from nodes $v_i$ with $i\le n-2$. Moreover, if at least $3$ nodes $v_i$ with $i\le n-2$ transmit simultaneously, \emph{no node} receives a message, including the potential communication between $v_{n-1}$ and $v_n$. This follows by observing that $P/R_B^\z > \Ic$: the signal power at the neighboring nodes must be more than the interference threshold, otherwise a node could receive two signals of the same power.

Assume $v_i$ with some $i\le n-2$ is the source. At the first round of the algorithm, nodes $v_1,v_2,\dots,v_{n-1}$ receive the message. Let $\mathcal{E}_t$ denote the event that $v_{n-1}$ transmits and no more than $s=3$ nodes $v_i$ with $i\le n-2$ transmit in round $t$. By the observations above, $v_n$ will not receive the message until $\mathcal{E}_t$ happens. Moreover,  during the subsequent steps of the algorithm, until $\mathcal{E}_t$ happens for the first time, all nodes at $v_1,v_2,\dots,v_{n-1}$ will have the same history and will be symmetric with respect to {\cd} and {\ack} primitives.

\begin{figure}[ht]
\centering
\subfloat[][non-spontaneous mode]{
\includegraphics[width=0.4\textwidth]{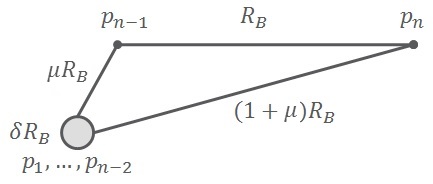}\label{fig:nonspon}
}
\qquad
\subfloat[][spontaneous mode]{
\includegraphics[width=0.45\textwidth]{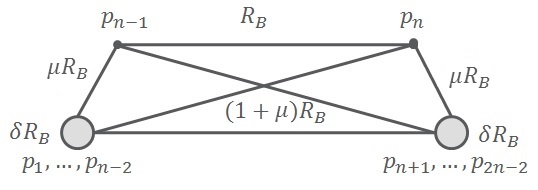}\label{fig:spon}
}
\caption{Distance diagrams of the lower bound instances.}
\end{figure}

It remains to bound the expected time  until $\mathcal{E}_t$ happens for the first time. We can assume w.l.o.g. that in each round there are no more than $s$ nodes transmitting. As discussed above, if $v_n$ is not informed in such a round, then the nodes that transmitted in the given round can learn at best that none of them is at $v_{n-1}$ and stop transmitting thereafter. Thus, we can assume that after each unsuccessful round there are at most $s+1$ nodes that stop transmitting and the probability $\mathbb{P}_r$ of success at round $0\le r \le n/(2s)$ (if no success occurred before round $r$) is at most
\[
\mathbb{P}_r\le \frac{\sum_{t=1}^{s-1}{n-sr\choose t}}{\sum_{t=1}^{s}{n-sr\choose t}}\le \frac{(s-1)\cdot{n-sr\choose s-1}}{(s-1)\cdot{n-sr\choose s-1} + {n-sr\choose s}}\le \frac{s^2}{n-(r+1)s +s^2}\le 2s^2/n,
\]
where the numerator counts the number subsets of an $n-sr$-element set that contain a fixed element (the node $v_{n-1}$) and no more than $s$ elements in total, and the denominator counts all subsets containing no more than $s$ elements.
A further straightforward calculation shows that the expected number of rounds until the first success is $\Omega(n)$ (see e.g.~\cite[Thm. 1]{DGKN13}).

In the case when the node $v_n$ might try to act without receiving a message, the instance above can be complemented as in  Fig. \ref{fig:spon}. A similar argument works in this case too.
\end{proof}

\end{document}